\documentclass[11pt,reqno,titlepage]{amsart}
\usepackage{amssymb,mathrsfs,color,comment,enumerate,hyperref,bbold,apptools}
\usepackage[all]{xy}
\usepackage{color}
\usepackage{amsaddr}
\usepackage{chngcntr}
\reversemarginpar
\AtAppendix{\counterwithin{theorem}{section}}
\AtAppendix{\counterwithin{equation}{section}}

\makeatletter
\renewcommand{\email}[2][]{
\ifx\emails\@empty\relax\else{\g@addto@macro\emails{,\space}}\fi
\@ifnotempty{#1}{\g@addto@macro\emails{\textrm{(#1)}\space}}
\g@addto@macro\emails{#2}
}
\makeatother

\numberwithin{equation}{section}

\theoremstyle{plain}
\newtheorem{theorem}{Theorem}[section]
\newtheorem{proposition}[theorem]{Proposition}
\newtheorem{corollary}[theorem]{Corollary}
\newtheorem{lemma}[theorem]{Lemma}

\theoremstyle{definition}
\newtheorem{assumption}[theorem]{Assumption}
\newtheorem{remark}[theorem]{Remark}
\newtheorem{example}[theorem]{Example}
\newtheorem{definition}[theorem]{Definition}

\newcommand{\cl}{\textnormal{cl}}
\newcommand{\R}{\mathbb{R}}      
\newcommand{\QW}{\mathbb{Q}}

\newcommand{\E}{\mathbb{E}}
\newcommand{\tn}{\textnormal}

\newcommand{\ind}{\textbf{1}}
\newcommand{\PW}{\mathbb{P}}
\newcommand{\Norm}{\|\cdot\|}
\newcommand{\N}{\mathbb{N}}
\newcommand{\dom}{\textnormal{dom}}
\newcommand{\eps}{\varepsilon}

\newcommand{\Borel}{\mathbb{B}}

\newcommand{\CF}{\mathcal F}

\newcommand{\CH}{\mathcal H}

\newcommand{\CB}{\mathcal B}
\newcommand{\CX}{\mathcal X}
\newcommand{\CP}{\mathcal P}
\newcommand{\CC}{\mathcal C}
\newcommand{\CS}{\mathcal S}
\newcommand{\CY}{\mathcal Y}
\newcommand{\CM}{\mathcal M}
\newcommand{\CL}{\mathcal L}
\newcommand{\acc}{\mathcal A}
\newcommand{\price}{\mathfrak p}
\newcommand{\Rhoi}{\rho_i}
\newcommand{\Rho}{\rho_{\mathcal R}}
\newcommand{\Lip}{\mathbf{C}}
\newcommand{\mf}{\mathfrak}
\newcommand{\alloc}{\mathbb A}
\newcommand{\peq}{\preceq}
\newcommand{\Linfty}{L^\infty}
\newcommand{\cost}{\mathfrak c}
\newcommand{\augm}{\mathscr A_++\textnormal{ker}(\pi)}
\newcommand{\two}{\twoheadrightarrow}
\newcommand{\mbf}{\mathbf}

\title[Risk Sharing with Multidimensional Security Markets]{Risk Sharing for Capital Requirements with Multidimensional Security Markets}

\author{Felix-Benedikt Liebrich\quad\quad\quad Gregor Svindland}
\address{Department of Mathematics, University of Munich, Germany}
\email{liebrich@math.lmu.de, svindla@math.lmu.de}
\date{September 25, 2018}


\begin{document}

\parindent 0em \noindent

\begin{abstract}
We consider the risk sharing problem for capital requirements induced by capital adequacy tests and security markets. The agents involved in the sharing procedure may be heterogeneous in that they apply varying capital adequacy tests and have access to different security markets. We discuss conditions under which there exists a representative agent. Thereafter, we study two frameworks of capital adequacy more closely, polyhedral constraints and distribution based constraints. We prove existence of optimal risk allocations and equilibria within these frameworks and elaborate on their robustness.  
\end{abstract}

\maketitle


\section{Introduction}

In this paper we consider the risk sharing problem for capital requirements. Optimal capital and risk allocation among economic agents, or business units, has for decades been a predominant subject in the
respective academic and industrial research areas. Measuring financial risks with capital requirements goes back to the seminal paper by Artzner et al.\ \cite{Artzner}. There, risk measures are \textit{by definition} capital requirements determined by two primitives: the \textit{acceptance set} and the \textit{security market}.

The acceptance set, a subset of an ambient space of losses, corresponds to a \textit{capital adequacy test}. A loss is deemed adequately capitalised if it belongs to the acceptance set, and inadequately capitalised otherwise. If a loss does not pass the capital adequacy test, the agent has to take prespecified remedial actions: she can raise capital in order to buy a portfolio of securities in the security market which, when combined with the loss profile in question, results in an adequately capitalised secured loss. 

Suppose the security market only consists of one num\'eraire asset, liquidly traded at arbitrary quantities. After discounting, one obtains a so-called \textit{monetary risk measure}, which is characterised by satisfying the \textit{cash-additivity property}, that is $\rho(X+a)=\rho(X)+a$. Here, $\rho$ denotes the monetary risk measure, $X$ is a loss, and $a\in \R$ is a capital amount which is added to or withdrawn from the loss. Monetary risk measures have been widely studied, see F\"ollmer \& Schied~\cite{FoeSch} and the references therein. As observed in Farkas et al. \cite{FKM2013, FKM2014, FKM2015} and Munari \cite{Munari}, there are good reasons for revisiting the original approach to risk measures of Artzner et al.\ \cite{Artzner}:\begin{itemize}
\item[(1)] Typically, more than one asset is available in the security market. It is also less costly for the agent to invest in a portfolio of securities designed to secure a specific loss rather than restricting the remedial action to investing in a single asset independent of the loss profile.
\item[(2)] Even if securitisation is constrained to buying a single asset, discounting with this asset may be impossible because it is not a num\'eraire; c.f. Farkas et al.\ \cite{FKM2014}.  Also, as risk is measured \textit{after discounting}, the discounting procedure is implicitly assumed not to add additional risk, which is questionable in view of risk factors such as uncertain future interest rates. For a thorough discussion of this issue see El Karoui \& Ravenelli \cite{ElKarouiRavanelli}. Often, risk is determined purely in terms of the \textit{distribution} of a risky position, a paradigm we discuss in detail below. Therefore, instability of this crucial \textit{law-invariance property} of a risk measure under discounting is another objection. If the security is not riskless (i.e.\ equals cash), losses which originally were identically distributed may not share the same distribution any longer after discounting, while losses that originally display different laws may become identically distributed.
\item[(3)] Without discounting, if only a single asset is available in the security market, cash-additivity requires the security to be riskless, and it is questionable whether such a security is realistically available, at least for longer time horizons. This is a particularly nagging issue in the insurance context. 
\end{itemize}

In this paper we will follow the original ideas in \cite{Artzner} and study the risk sharing problem for risk measures induced by acceptance sets and possibly multidimensional security spaces. We consider a one-period market  populated by a finite number $n\geq 2$ of agents who seek to secure losses occurring at a fixed future date, say tomorrow. We attribute to each agent $i\in\{1,...,n\}$ an ordered vector space $\CX_i$ of \textit{losses net of gains} she may incur, an acceptance set $\acc_i\subset\CX_i$ as capital adequacy test, and a security market consisting of a subspace $\CS_i\subset\CX_i$ of security portfolios as well as observable prices of these securities given by a linear functional $\price_i:\CS_i\to\R$. As the securities in $\CS_i$ are deemed suited for hedging, the linearity assumptions on $\CS_i$ and $\price_i$ reflects that they are liquidly traded and their bid-ask spread is zero. The risk attitudes of agent $i$ are fully reflected by the resulting \textit{risk measure}
\begin{equation}\label{eq:RMprelim}\Rhoi(X):=\inf\{\price_i(Z)\mid Z\in\CS_i, X-Z\in\acc_i\},\quad X\in\CX_i,\end{equation}
that is the minimal capital required to secure $X$ with securities in $\CS_i$.

The problem we consider is how to reduce the aggregated risk in the system by means of redistribution. Formally, given a market loss $X\in\sum_{i=1}^n\CX_i$ the system in total incurs, we need to solve the optimisation problem
\begin{equation}\label{eq:risksharing}\sum_{i=1}^n\Rhoi(X_i)\to\min\quad\tn{subject to~}X_i\in\CX_i\tn{~and~}X_1+...+X_n=X.\end{equation}
A vector $\mathbf X=(X_1,...,X_n)$, a so-called allocation of $X$, which solves the optimisation problem and yields a finite optimal value is \textit{Pareto optimal}. However, this resembles  centralised redistribution which attributes to each agent a certain portion of the aggregate loss in an overall optimal way without considering individual well-being. Redistribution by agents trading portions of losses at a certain price while adhering to individual rationality constraints leads to the notion of \textit{equilibrium allocations} and \textit{equilibrium prices}, a variant of the risk sharing problem above. 

Special instances of this general problem have been extensively studied in the literature. Borch \cite{Borch}, Arrow \cite{Arrow} and Wilson \cite{Wilson} consider the problem for expected utilities. More recent are studies for convex monetary risk measures, starting with Barrieu \& El Karoui \cite{ElKaroui} and Filipovi\'c \& Kupper \cite{FilKup}. A key assumption which allows to prove existence of optimal risk sharing for convex monetary risk measures is \textit{law-invariance}, i.e. the measured risk is the same for all losses which share the same distribution under a benchmark probability model, see Jouini et al.\ \cite{JST}, Acciaio \cite{Fatou}, Acciaio \& Svindland \cite{Acciaio}, and Filipovi\'c \& Svindland \cite{Svindland08}. For a thorough discussion of the existing literature, we refer to Embrechts et al.\ \cite{Embrechts}.

\textbf{Main contributions.} In the following we summarise the four main contributions of this paper. 

First of all, we prove a representative agent formulation of the risk sharing problem: the behaviour of the interacting agents in the market is, under mild assumptions, captured by a \textit{market capital requirement} of type \eqref{eq:RMprelim}, namely 
\begin{center}$\Lambda(X)=\inf\{\pi(Z)\mid Z\in\CM,~X-Z\in\mathscr A_+\}$,\end{center} 
where $\Lambda(X)$ is the infimal level of aggregated risk realised by redistribution of $X$ as in \eqref{eq:risksharing}, $\mathscr A_+$ is a market acceptance set, and $(\CM,\pi)$ is a global  security market.

Secondly, we study two prominent cases, mostly characterised by the involved notions of acceptability, for which we prove that the risk sharing problem \eqref{eq:risksharing}, including the quest for equilibria,  admits solutions. In the first instance, individual losses are --- in the widest sense --- contingent on scenarios of the future state of the economy. A loss is deemed acceptable if certain capital thresholds are not exceeded under a fixed finite set of linear aggregation rules which may vary from agent to agent. The reader may think of a combination of finitely many valuation and stress test rules as studied in Carr et al. \cite{Madan}, see also \cite[Section 4.8]{FoeSch}. The resulting acceptance sets will thus be \textit{polyhedral}. In the second class of acceptance sets under consideration, whether or not a loss is deemed adequately capitalised only depends on its distributional properties under a fixed reference probability measure, not on scenariowise considerations: acceptability is a statistical notion. More precisely, losses are modelled as random variables on a probability space $(\Omega,\CF,\PW)$, and the respective individual acceptance sets $\acc_i$, $i\in\{1,...,n\}$, will be \textit{law-invariant}. However, in non-trivial cases the security spaces $\mathcal S_i$ will not be law-invariant, hence securitisation depends on the potentially varying joint distribution of the loss and the security and is thus statewise rather than distributional. This both reflects the practitioner's reality and is mathematically interesting as the resulting capital requirements $\Rhoi$ are far from law-invariant. Only if the security spaces are spanned by the cash asset, i.e.\ $\CS_i=\R$, law-invariance of the acceptance set implies law-invariance of the corresponding monetary risk measure. In that case the risk sharing problem has been solved, c.f. \cite{Svindland08, JST}. We will utilise these results, but like to emphasise that reducing the general problem to the law-invariant cash-additive case is impossible.

As a third contribution, we not only prove the existence of solutions, i.e.\ of \textit{optimal risk allocations}, but carefully study continuity properties of the set-valued map assigning to an aggregated loss its optimal risk allocations. These reflect that optimality is robust under misspecification of the input, the aggregated loss profile. If the map is \textit{upper hemicontinuous}, a slight miscalculation of the aggregated loss does not change the set of such optimal risk allocations drastically. It is also useful from a constructive point of view: one can solve the risk sharing problem of a complex loss by approximating it with simpler losses and calculating solutions of these instead. \textit{Lower hemicontinuity}, on the other hand, guarantees that a given optimal risk allocation stays close to optimal under a slight perturbation of the underlying aggregated loss. 

At last, we study optimal splitting problems in the spirit of Tsanakas \cite{Tsanakas} and Wang \cite{Wang} as an application of the general theory: Under the presence of market frictions such as transaction costs, can a financial institution split an aggregated loss optimally by introducing subsidiaries subject to potentially varying regulatory regimes and having access to potentially varying security markets?

\textbf{Structure of the paper.} In Section~\ref{sec:model} we rigorously introduce risk measurement in terms of capital requirements, agent systems, optimal allocations, and equilibria.  
Section~\ref{sec:maths} presents the representative agent formulation of the risk sharing problem and proves useful meta results. These are key to the discussion of risk sharing involving polyhedral acceptance sets in Section~\ref{sec:polyagent} and law-invariant acceptance sets in Section \ref{sec:lawinvacc}, as well as optimal portfolio splits in Section \ref{sec:regarbitrage}.  
Technical supplements are relegated to the appendix.

\section{Agent systems and optimal allocations}\label{sec:model}

\subsection{Preliminaries}

In a first step of modelling, we assume that the attitude of individual agents towards risk is given by a \textit{risk measurement regime} and corresponding {\em risk measure}. 
\begin{definition}\label{def:RM}Let $(\CX,\peq)$ be an ordered vector space, $\CX_+:=\{X\in\CX\mid 0\peq X\}$ be its positive cone, and $\CX_{++}:=\CX_+\backslash\{0\}$. 
\begin{itemize}
\item An \textsc{acceptance set} is a non-empty proper and convex subset $\acc$ of $\CX$ which is monotone, i.e. $\acc-\CX_+\subseteq \acc$.\footnote{ Here and in the following, given subsets $A$ and $B$ of a vector space $\CX$, $A+B$ denotes their Minkowski sum $\{a+b\mid a\in A, b\in B\}$, and $A-B:=A+(-B)$.}
\item A \textsc{security market} is a pair $(\CS,\price)$ consisting of a finite-dimensional linear subspace $\CS\subset\CX$ and a positive linear functional $\price:\CS\to\R$ such that there is $U\in\CS\cap\CX_{++}$ with $\price(U)=1$. The elements $Z\in\mathcal S$ are called security portfolios or simply securities, and $\CS$ is the \textsc{security space}, whereas $\price$ is called \textsc{pricing functional}.
\item A triple $\mathcal R:=(\acc, \CS, \price)$ is a \textsc{risk measurement regime} if $\acc$ is an acceptance set and $(\mathcal S,\price)$ is a security market such that the following no-arbitrage condition holds:
\begin{equation}\label{eq:regime}\forall\,X\in \CX:\quad\sup\{\price(Z)\mid Z\in\CS, X+Z\in\acc\}<\infty.\end{equation}
\item The \textsc{risk measure} associated to a risk measurement regime $\mathcal R$ is the functional 
\begin{equation}\label{defRM}\Rho: \CX\rightarrow(-\infty,\infty],\quad X\mapsto\inf\left\{\price(Z)\mid Z\in\mathcal S, X-Z\in\mathcal A\right\}. \end{equation}
\end{itemize}
$\Rho$ is \textsc{normalised} if $\Rho(0)=0$, or equivalently $\sup_{Z\in\acc\cap\mathcal S}\price(Z)=0$. It is \textsc{lower semicontinuous} (l.s.c.)\,with respect to some vector space topology $\tau$ on $\CX$ provided every lower level set $\{X\in\CX\mid \Rho(X)\leq c\}$, $c\in\R$, is $\tau$-closed.
\end{definition}

Immediate consequences of the definition of $\Rho$ are the following properties:

\begin{itemize}
\item $\Rho$ is a proper function\footnote{ Given a non-empty set $M$, a function $f:M\to[-\infty,\infty]$ is proper if $f^{-1}(\{-\infty\})=\emptyset$ and $f\not\equiv \infty$.} by \eqref{eq:regime} and $\Rho(Y)\leq 0$ for any choice of $Y\in\acc$. Moreover, it is convex, i.e.\ $\Rho(\lambda X+ (1-\lambda)Y)\leq \lambda \Rho(X)+ (1-\lambda)\Rho(Y)$ holds for all choices of $\lambda\in [0,1]$ and $X,Y\in \CX$;
\item  \textsc{$\peq$-monotonicity}, i.e.\ $X\peq Y$ implies $\Rho(X)\leq \Rho(Y)$;
\item  \textsc{$\CS$-additivity}, i.e\ $\Rho(X+Z)=\Rho(X)+\price(Z)$ for all $X\in \CX$ and all $Z\in\CS$.
\end{itemize}

Note that risk measures as in \eqref{defRM} evaluate the risk of \textit{losses net of gains} $X\in \CX$. The positive cone $\CX_+$ corresponds to pure losses. Therefore, $\Rho$ is non-decreasing with respect to $\preceq$, not non-increasing as in most of the literature on risk measures where the risk of \textit{gains net of losses} is measured. In the security market, however, we consider the usual monotonicity, i.e. a security $Z^*\in\CS$ is better than $Z\in\CS$ if $Z\peq Z^*$. This also explains positivity of the pricing functional $\price:\CS\to\R$. Combining these two viewpoints, the impact of a security $Z\in\CS$ on a loss profile $X\in\CS$ is given by $X-Z$, and $\Rho(X)$ is the infimal price that has to be paid for a security $Z$ in the security market with loss profile $-Z$ in order to reduce the risk of $X$ to an acceptable level. The no-arbitrage condition \eqref{eq:regime} means that one cannot short arbitrary valuable securities and stay acceptable. 

There is a close connection between capital requirements defined by \eqref{defRM} and superhedging. Given a risk measurement regime $\mathcal R=(\acc,\CS,\price)$ on an ordered vector space $(\CX,\peq)$, let $\ker(\price):=\{N\in\CS\mid \price(N)=0\}$ denote the kernel of the pricing functional, i.e. the set of fully leveraged security portfolios available at zero cost. Moreover, fix an arbitrary $U\in\CS\cap\CX_+$ such that $\price(U)=1$. Each $Z\in\CS$ can be written as $Z=\price(Z)U+(Z-\price(Z)U)$, and $Z-\price(Z)U\in\ker(\price)$. Hence,  $X\in\CX$ and $Z\in\CS$ satisfy $X-Z\in\acc$ if, and only if, for $r:=\price(Z)\in\R$ we can find $N\in\ker(\price)$ such that $rU+N+(-X)\in-\acc$. The risk $\Rho(X)$ may thus be expressed as 
\begin{align*}\Rho(X)&=\inf\{\price(Z)\mid Z\in\CS, X-Z\in\acc\}\\
&=\inf\{r\in\R\mid \exists\, N\in\ker(\price):~N+rU+(-X)\in-\acc\},
\end{align*}
The set $-\acc$ is the set of acceptable \textit{gains net of losses}, and $-X$ is the payoff associated to the loss profile $X$. The elements in $\ker(\price)$ are \textit{zero cost investment opportunities}. If we conservatively choose the acceptance set $\acc=-\CX_+$,
\[\Rho(X)=\inf\{r\in\R\mid \exists\, N\in\ker(\price):~N+rU+(-X)\succeq 0\},\]
that is we recover by $\Rho(X)$ the superhedging price of the payoff $-X$. A general risk measurement regime thus leads to a superhedging functional involving the relaxed notion of superhedging $N+rU-X\in-\acc$. In the terminology of superhedging theory, $\Rho(X)$ is the infimal amount of cash that needs to be invested in the security $U$ such that $X$ can be superhedged when combined with a suitable zero cost trade in the (security) market. Such relaxed superhedging functionals have been recently studied by, e.g., Cheridito et al. \cite{Cheridito}. The separation between $U$ and $\ker(\price)$ introduced above will be useful throughout the paper. 

\subsection{Agent systems}
In order to introduce the risk sharing problem in precise terms, a notion of the interplay of the individual agents and their respective capital requirements is required; for terminology concerning ordered vector spaces, we refer to \cite[Chapters 8--9]{Ali}. We consider an abstract one-period market which incurs aggregated losses net of gains modelled by a Riesz space $(\CX,\peq)$. The market is comprised of $n\geq 2$ agents, and throughout the paper we identify each individual agent with a natural number $i$ in the set $\{1,...,n\}$, which we shall denote by $[n]$ for the sake of brevity.
The agents might have rather heterogeneous assessments of risks. This is firstly reflected by the assumption that each agent operates on an \textit{(order) ideal}\footnote{ An ideal $\mathcal Y$ of a Riesz space $(\CX,\peq)$ is a subspace in which the inclusion $\{Z\in\CX\mid |Z|\peq |Y|\}\subset\CY$ holds for all $Y\in\mathcal Y$.}  $\CX_i\subset\CX$, $i\in[n]$, which may be a proper subset of $\CX$. Without loss of generality we shall impose $\CX=\CX_1+...+\CX_n$. Within each ideal, and thus for each agent, adequately capitalised losses are encoded by an acceptance set $\acc_i\subset\CX_i$. Agent $i\in[n]$ is allowed to secure losses she may incur with securities from a security market $(\CS_i,\price_i)$, where $\CS_i\subsetneq\CX_i$. We shall impose that each $\mathcal R_i:=(\acc_i,\CS_i,\price_i)$ is a risk measurement regime on $(\CX_i,\peq|_{\CX_i\times\CX_i})$, $i\in[n]$. 
In sum, the individual risk assessments are fully captured by the $n$-tuple of risk measurement regimes $(\mathcal R_1,...,\mathcal R_n)$.
\begin{definition}\label{def:agentsystem}
An $n$-tuple $(\mathcal R_1,...,\mathcal R_n)$, where, for each $i\in[n]$, $\mathcal R_i$ is a risk measurement regime on $\CX_i$, is called an \textsc{agent system} if 
\begin{itemize}
\item[($\star$)] For all $i,j\in[n]$, the pricing functionals $\price_i$ and $\price_j$ agree on $\CS_i\cap \CS_j$. Moreover, if we set $i\sim j$ if $i\neq j$ and $\price_i$ is non-trivial on $\CS_i\cap\CS_{j}$, the resulting graph 
\begin{center}$G=([n],\{\{i,j\}\subset [n]\mid i\sim j\})$\end{center}
is connected. 
\end{itemize}
\end{definition}

Axiom ($\star$) clarifies the nature of the interaction of the involved agents: prices for securities accepted by more than one agent have to agree, and any two agents may interact and exchange securities by potentially  invoking other agents as intermediaries. Throughout this paper we will assume that the agents $[n]$ form an agent system. Such a situation is not too far-fetched: 
\begin{definition} The space of \textsc{jointly accepted securities} is $\check{\CS}:=\bigcap_{i=1}^n\CS_i$, whereas $\mathcal M :=\CS_1+...+\CS_n$ is the \textsc{global security space}.
\end{definition}
If, besides agreement of prices, $\check{\CS}\neq\{0\}$ and $\price_i|_{\check{\CS}}\neq 0$ for some and thus all $i\in[n]$, then assumption ($\star$) is met. The resulting graph is the complete graph on $n$ vertices. Moreover, if all agents operate on one and the same space $\CX_i=\CX$, $i\in[n]$, and the available security markets are unanimous and given by $\CS_i=\R\cdot U$, $i\in[n]$, for some $U\in\CX_{++}$ and $\price_i(rU)=r$, $r\in\R$, $(\mathcal R_1,...,\mathcal R_n)$ is an agent system. If we further specify $\CX$ to be a sufficiently rich space of random variables and $U=\ind$ is the constant random variable with value 1, the results for risk sharing with convex monetary risk measures can be embedded in our setting of agent systems; c.f. \cite{Fatou,Acciaio,Svindland08,JST}.

In the following we write $\Rhoi$ instead of $\rho_{\mathcal R_i}$ for the sake of brevity. Aggregated losses in $\CX$ will be denoted by $X$ or $W$, securities by $Z$, $U$ or $N$ throughout the paper. In order to introduce the risk sharing associated with $(\mathcal R_1,...,\mathcal R_n)$, we need the notion of attainable and security allocations:
\begin{definition}\label{def:allocation} A vector $\mbf X=(X_1,...,X_n)\in\prod_{i=1}^n\CX_i$ is an \textsc{attainable allocation} of an aggregated loss $W\in\CX$ if $W=X_1+...+X_n$. We denote the set of all attainable allocations of $W\in\CX$ by $\alloc_W$. \\
Given a global security $Z\in\CM$, we denote by $\alloc_Z^s:=\alloc_Z\cap\prod_{i=1}^n\CS_i$ the set of \textsc{security allocations} of $Z$. 
\end{definition}

\subsection{The risk sharing problem and its solutions}

Given a set $S\neq\emptyset$ and a function $f:S\to[-\infty,\infty]$, we set $\dom(f):=\{s\in S\mid f(s)<\infty\}$ to be the \textsc{effective domain} of $f$. We will also abbreviate its lower level sets by $\mathcal L_c(f):=\{s\in S\mid f(s)\le c\}$, $c\in\R$. 

We are now prepared to introduce the risk sharing problem. Its objective is to minimise the aggregated risk within the system. The allowed remedial action is reallocating an aggregated loss $W\in \CX$ among the agents involved:
\begin{equation}\label{eq:RS1}\sum_{i=1}^n\Rhoi(X_i)\to\min\quad\tn{subject to }\mbf X\in\alloc_W.\end{equation}
Obviously, the optimal value in \eqref{eq:RS1} is less than $+\infty$ if, and only if, $W\in\sum_{i=1}^n\dom(\Rhoi)$. It is furthermore well-known that \eqref{eq:RS1} is closely related to certain notions of economically optimal allocations which we define in the following.

\begin{definition}\label{def:Pareto} Let $(\mathcal R_1,...,\mathcal R_n)$ be an agent system on an ordered vector space $(\CX,\peq)$, let $W\in\CX$ be an aggregated loss, and let $\mbf W\in\prod_{i=1}^n\CX_i$ be a vector of initial loss endowments. 
\begin{enumerate}[(1)]
\item An attainable allocation $\mbf X\in\alloc_W$ is \textsc{Pareto optimal} if $\Rhoi(X_i)<\infty$, $i\in [n]$, and for any $\mbf Y\in\alloc_W$ with the property $\Rhoi(Y_i)\leq \Rhoi(X_i)$, $i\in[n]$, in fact $\Rhoi(X_i)=\Rhoi(Y_i)$ has to hold for all $i\in[n]$.
\item Suppose $\CX$ additionally carries a vector space topology $\tau$ such that $(\CX,\peq,\tau)$ is a topological Riesz space. A tuple $(\mbf X,\phi)$ is an \textsc{equilibrium} of $\mbf W$ if 
\begin{itemize}
\item $\mbf X\in\alloc_{W_1+...+W_n}$, 
\item $\phi\in\CX^*$ is positive with $\phi|_{\CS_i}=\price_i$, $i\in [n]$,
\item the budget constraints $\phi(-X_i)\leq \phi(-W_i)$, $i\in [n]$, hold,\footnote{ Note that the minus sign in the budget constraints is due to the fact that the elements in $\CX_i$ model losses, whereas $\phi$ prices payoffs.}  
\item and $\Rhoi(X_i)=\inf\{\Rhoi(Y)\mid Y\in\CX_i,~\phi(-Y)\leq \phi(-W_i)\}$ for all $i\in[n]$. 
\end{itemize} In that case, $\mbf X$ is called \textsc{equilibrium allocation} and $\phi$ \textsc{equilibrium price}.
\end{enumerate}
\end{definition}

Our first result links Pareto optima, equilibria, and solutions to the risk sharing problem \eqref{eq:RS1}. Proposition~\ref{prop:equivalence}(2) is indeed the First Fundamental Theorem of Welfare Economics adapted to our setting. 
\begin{proposition}\label{prop:equivalence} Under the assumptions of Definition \ref{def:Pareto}, the following holds:
\begin{enumerate}
\item[\tn{(1)}] If $W\in\sum_{i=1}^n\dom(\Rhoi)$, then $\mbf X\in\alloc_W$ is a Pareto optimal attainable allocation of $W$ if, and only if, 
\begin{equation}\label{eq:exact}\sum_{i=1}^n\Rhoi(X_i)=\inf_{\mbf Y\in\alloc_W}\sum_{i=1}^n\Rhoi(Y_i).\end{equation}
\item[\tn{(2)}] Any equilibrium allocation is Pareto optimal. 
\end{enumerate}
\end{proposition}

The proof requires the following well-known characterisation of Pareto optima; see, e.g., \cite[Proposition 3.2]{Ravanelli}.
\begin{lemma}\label{lem:characPareto}
Let $W\in\sum_{i=1}^n\dom(\Rhoi)$. If $\mbf X$ is a Pareto optimal attainable allocation of $W$, there are so-called Negishi weights $\lambda_i\geq 0$, $i\in[n]$, not all equal to zero, such that 
\begin{equation}\label{eq:Pareto}\sum_{i=1}^n\lambda_i\Rhoi(X_i)=\inf_{\mbf Y\in\alloc_W}\sum_{i=1}^n\lambda_i\Rhoi(Y_i).\end{equation}
Conversely, if $\mbf X\in\alloc_X$ satisfies \eqref{eq:Pareto} for a set of strictly positive weights $\lambda_i>0$, $i\in[n]$, then $\mbf X$ is a Pareto optimal attainable allocation.
\end{lemma}

The proof of Proposition \ref{prop:equivalence} shows that the agent system property ($\star$) dictates the values of the Negishi weights.

\begin{proof}[Proof of Proposition \ref{prop:equivalence}]
(1) By Lemma~\ref{lem:characPareto} any solution to \eqref{eq:exact} is Pareto optimal. Conversely, let $W\in\sum_{i=1}^n\dom(\Rhoi)$ and let $\mbf X\in\alloc_W$ be a Pareto optimal attainable allocation. 
Let $\lambda\in\R^n_{++}$ be any vector of Negishi weights such that $\mbf X$ is a solution to \eqref{eq:Pareto}. 
Recall the symmetric relation $\sim$ in ($\star$) and consider $j,k\in[n]$ such that $j\sim k$. By definition, we may find $Z\in\CS_j\cap\CS_k$ such that $p:=\price_j(Z)=\price_k(Z)\neq 0$. For $t\in\R$ let 
\begin{center}$\mbf X^t:=\mbf X+\frac{t}{p}Ze_j-\frac t pZe_k\in\alloc_X.$\end{center}
Here, $Ze_j$ is the vector whose $j$-th entry is $Z$, whereas all other entries are $0$. Analogously, we define $Ze_k$. By $\CS_i$-additivity of all $\Rhoi$'s, we infer
$$-\infty<\sum_{i=1}^n\lambda_i\Rhoi(X_i)\leq\inf_{t\in\R}\sum_{i=1}^n\lambda_i\Rhoi(X_i^t)=\sum_{i=1}^n\lambda_i\Rhoi(X_i)+\inf_{t\in\R}t(\lambda_j-\lambda_k).$$
This is only possible if $\lambda_j=\lambda_{k}$. Using that the graph $G$ in ($\star$) is connected, one inductively shows $\lambda_1=...=\lambda_n$. Dividing both sides of \eqref{eq:Pareto} by $\lambda_1$ yields
\[\sum_{i=1}^n\Rhoi(X_i)=\inf_{\mbf Y\in\alloc_X}\sum_{i=1}^n\Rhoi(Y_i).\]

(2) Suppose that $\mbf W$ is an initial loss endowment with associated equilibrium $(\mbf X,\phi)$. Then $\phi(X_i)=\phi(W_i)$ holds by monotonicity. Given $Z_i\in \CS_i$ such that $\price_i(Z_i)=1$ and arbitrary $Y_i\in \CX_i$,
\[\phi(Y_i+(\phi(X_i)-\phi(Y_i))Z_i)=\phi(X_i)=\phi(W_i)\] 
holds as $\phi=\price_i$ on $\CS_i$. Thus the budget constraint is satisfied, and hence $$\Rhoi(X_i)\leq \Rhoi(Y_i + \phi( X_i-Y_i)Z_i) = \Rhoi(Y_i)+ \phi( X_i) -\phi(Y_i).$$ If we set $W:=W_1+...+W_n$, for any other allocation $\mbf Y\in \alloc_W$ we obtain $$\sum_{i=1}^n \Rhoi(X_i)\leq \sum_{i=1}^n \Rhoi(Y_i)+ \phi( X_i)-\phi(Y_i)= \sum_{i=1}^n \Rhoi(Y_i) $$ since $\sum_{i=1}^n\phi( X_i) = \sum_{i=1}^n\phi( Y_i)= \phi(W)$. By (1), $\mbf X$ is Pareto optimal.
\end{proof}


\section{Infimal convolutions and the representative agent}\label{sec:maths}

This section is comprised of the formal mathematical treatment of the risk sharing problem on ideals of a Riesz space as introduced in Section~\ref{sec:model}. We shall link risk sharing to the infimal convolution of the individual risk measures, prove its representation as a capital requirement for the market, i.e.\ a \textit{representative agent}, and find powerful sufficient conditions for its solvability. 

Beforehand, however, we need to introduce further axioms that an agent system may satisfy in addition to ($\star$). We shall refer to them at various stages of the paper, they are however not assumed to be met throughout. For $n\geq 2$ let $(\mathcal R_1,...,\mathcal R_n)$ be an agent system.
\begin{itemize}
\item[(NSA)]\textsc{No security arbitrage:} For some $j\in[n]$ it holds that 
\begin{center}$\left(\sum_{i\neq j}\ker(\price_i)\right)\cap\CS_j\subset\ker(\price_j)$;\end{center}
\item[(NR)]\textsc{Non-redundance of the joint security market:} There is $Z\in\check{\CS}$ and $\mbf Z\in\alloc_Z^s$ such that $\sum_{i=1}^n\price_i(Z_i)\neq 0.$
\item[(SUP)] \textsc{Supportability:} Given a locally convex Hausdorff topological Riesz space $(\CX,\preceq,\tau)$ with dual space $\CX^*$, there is some $\phi_0\in\CX^*_+$ and a constant $\gamma\in\R$ such that
\begin{enumerate}[(i)]
\item for all $\mbf Z\in\prod_{i=1}^n\CS_i$ with $\sum_{i=1}^n\price_i(Z_i)=0$ we have $\phi_0(Z_1+...+ Z_n)=0$ and for some $\tilde{\mbf Z}\in\prod_{i=1}^n\CS_i$ with $\sum_{i=1}^n\price_i(\tilde Z_i)\neq 0$ we have $\phi_0(\tilde Z_1+...+\tilde Z_n)\neq 0$;
\item for all $\mbf Y\in\prod_{i=1}^n\acc_i$ we have $\phi_0(Y_1+...+Y_n)\leq \gamma$.
\end{enumerate}
\item[(SUP$_\infty$)] \textsc{Infinite supportability:} $(\mathcal R_i)_{i\in\N}$ is a sequence of risk measurement regimes on one and the same locally convex Hausdorff topological Riesz space $(\CX,\preceq,\tau)$ such that $(\mathcal R_1,...,\mathcal R_n)$ satisfies ($\star$) for all $n\in\N$ and such that there is some $\phi_0\in\CX^*$ with 
\begin{center}$\sum_{i\in\N}\sup_{Y\in\acc_i}\phi_0(Y)<\infty\text{ and }\phi_0|_{\CS_i}=\price_i,~i\in\N.$\end{center}
\end{itemize}

Condition (NSA) is violated if each agent would be able to obtain arbitrarily valuable securitisation from the other agents, who can provide it at zero cost. That would reveal a mismatch of security markets leading to hypothetical infinite wealth for all agents. Non-redundance of the joint security market is in particular satisfied if there is $Z\in \check{\CS}$ such that $\price_i(Z)\neq 0$ for some $i\in [n]$, and thus by the defining property ($\star$) of an agent system for all $i\in[n]$. Hence, under (NR) there is a jointly accepted security valuable for the market. As regards condition (SUP), think of $\phi_0$ as a pricing functional. (i) is a consistency requirement between $\phi_0$ and the individual prices $\price_i$. (ii) reads as the impossibility to decompose a loss $X$ acceptably for all agents if $X$ is sufficiently poor, that is the value $\phi_0(-X)$ of the corresponding payoff $-X$ under $\phi_0$ is less than a certain level $-\gamma$. (SUP$_\infty$) is a strengthening of (SUP) for all finite subsystems of $(\mathcal R_i)_{i\in\N}$.

\subsection{The risk sharing functional and the representative agent} 

Proposition \ref{prop:equivalence} motivates the definition of the \textsc{risk sharing functional}:
\[\Lambda:\CX\to[-\infty,\infty],\quad X\mapsto\inf_{\mbf X\in\alloc_X}\sum_{i=1}^n\Rhoi(X_i).\]
It corresponds to the so-called \textsc{infimal convolution} of the risk measures $\rho_1,...,\rho_n$, and thus inherits properties like $\peq$-monotonicity and convexity. We refer to Appendix~\ref{sec:infimal:convolution}, in particular Lemma~\ref{lem:properties}, for a brief summary of these facts.

Our next result implies that, if proper, $\Lambda$ is again a risk measure of type \eqref{defRM}: the shared risk level is the minimal price the market has to pay for a cumulated security that renders market acceptability. Thus, market behaviour may be seen as the behaviour of a \textit{representative agent} operating on $\CX$. Recall from Definition \ref{def:allocation} that $\alloc_Z^s$ denotes the set of security allocations of $Z\in\CM$.

\begin{proposition}\label{prop:operational} Define $\pi(Z):=\inf_{\mbf Z\in\alloc_Z^s}\sum_{i=1}^n\price_i(Z_i)$, $Z\in\CM$.
\begin{enumerate}
\item[\tn{(1)}] For any $Z\in\CM$ and arbitrary $\mbf Z\in\alloc_Z^s$, $\pi(Z)$ may be represented as
\[\pi(Z)=\sum_{i=1}^n\price_i(Z_i)+\pi(0).\]
Either $\pi(0)=0$ or $\pi(0)=-\infty$. $\pi(0)=0$ is equivalent to \tn{(NSA)}, and in that case $\pi$ is real-valued, linear, and satisfies $\pi|_{\CS_i}=\price_i$, $i\in [n]$. Otherwise $\pi\equiv -\infty$.
\item[\tn{(2)}] $\Lambda$ can be represented as 
\[\Lambda(X)=\inf\left\{\pi(Z)\mid Z\in\mathcal M,~X-Z\in\mathscr B\right\},\quad X\in\CX,\]
for any monotone and convex set $\mathscr B\subset\CX$ satisfying $\mathscr A_+\subset\mathscr B\subset\mathcal L_0(\Lambda)$. Here,  
\begin{center}$\mathscr A_{+}:=\sum_{i=1}^n\acc_i.$\end{center}
\item[\tn{(3)}] If \tn{(NSA)} and \tn{(SUP)} hold, $\Lambda$ is proper.
\item[\tn{(4)}] If $\Lambda$ is proper, then \tn{(NSA)} holds, i.e. $\pi(0)=0$, and $\pi$ is positive. In that case, $(\mathscr A_+,\CM,\pi)$ is a risk measurement regime on $\CX$ and $\Lambda$ is the associated risk measure.
\end{enumerate}
\end{proposition}
\begin{proof}
(1) Let $Z\in\mathcal M$ and let $\mbf Z\in\alloc_Z^s$ be arbitrary, but fixed. The identity $\alloc_Z^s=\mbf Z+\alloc_0^s$ implies
\[\pi(Z)=\sum_{i=1}^n\price_i(Z_i)+\inf_{\mbf N\in\alloc_0^s}\sum_{i=1}^n\price_i(N_i)=\sum_{i=1}^n\price_i(Z_i)+\pi(0).\]
Consider $\mathcal V:=\{(\price_i(N_i))_{i\in [n]}\mid\mbf N\in\alloc_0^s\}$, which is a subspace of $\R^n$. In the following, we denote by $e_l$ the $l$-th unit vector of $\R^n$. Note that $\pi(0)=0$ if, and only if, dim$(\mathcal V)<n$. Indeed, let  $\ind=(1,1,...,1)\in\R^n$ and observe that $\pi(0)=\inf_{x\in \mathcal V}\langle \ind, x \rangle$ which is $-\infty$ in case dim$(\mathcal V)=n$. Suppose that dim$(\mathcal V)<n$, i.e. $\mathcal V^\perp\neq \{0\}$, and let $0\neq \lambda\in\mathcal V^\perp$. As in the proof of Proposition \ref{prop:equivalence}(1), $e_j-e_{k}\in\mathcal V$ holds for all $j,k\in[n]$ such that $j\sim k$, which implies $\lambda_j=\lambda_k$. As the relation $\sim$ induces a connected graph, $\lambda\in\R\cdot\ind=\mathcal V^\perp$. Hence, we obtain that  $\langle \ind, x \rangle=0$ for all $x\in \mathcal V$ which implies $\pi(0)=0$, so we have proved equivalence of $\pi(0)=0$ and dim$(\mathcal V)<n$.  But dim$(\mathcal V)<n$ is equivalent to the fact that there is a $j\in[n]$ such that $e_j\notin\mathcal V$, which in turn is equivalent to (NSA): whenever $Z\in\CS_j$ lies in the Minkowski sum $\sum_{i\neq j}\ker(\price_i)$, $\price_j(Z)=0$ has to hold.

(2) We first note that $\mathscr A_+$ is convex and monotone. Indeed, let $X,Y\in\CX$ such that $Y\in\mathscr A_+$ and $X\peq Y$. Fix $\mbf Y\in\alloc_Y$ such that $Y_i\in\acc_i$, $i\in[n]$, and $\mbf X\in\alloc_X$ arbitrary. By the Riesz Decomposition Property (c.f. \cite[Section 8.5]{Ali}), there are $W_1,...,W_n\in \CX_+$ such that $Y-X=\sum_{i=1}^n W_i$ and $W_i\peq |Y_i-X_i|$, which means $\mbf W\in\alloc_{Y-X}$. Hence, for all $i\in[n]$, we obtain $Y_i-W_i\in\acc_i$ by monotonicity of $\acc_i$, and thus $X=\sum_{i=1}^nY_i-W_i\in\mathscr A_+$. Moreover, $\CL_0(\Lambda)$ is monotone and convex as well, which follows from the corresponding properties of $\Lambda$.

For $\mathscr B\subset\mathscr B'$, we have  
\[\inf\left\{\pi(Z)\mid Z\in\mathcal M,~X-Z\in\mathscr B\right\}\geq \inf\left\{\pi(Z)\mid Z\in\mathcal M,~X-Z\in\mathscr B'\right\}.\]
As $\mathscr A_+\subset \CL_0(\Lambda)$, (2) is proved if for arbitrary $X\in\CX$ we can show the two estimates
\begin{equation}\label{eq:help1}\Lambda(X)\geq \inf\left\{\pi(Z)\mid Z\in\mathcal M,~X-Z\in\mathscr A_+\right\},\end{equation}
and 
\begin{equation}\label{eq:help2}\Lambda(X)\leq \inf\left\{\pi(Z)\mid Z\in\mathcal M,~X-Z\in\CL_0(\Lambda)\right\}.\end{equation}
The first assertion trivially holds if $\Lambda(X)=\infty$. If $X\in\dom(\Lambda)=\sum_{i=1}^n\dom(\Rhoi)$, choose $\mbf X\in\alloc_X$ such that $\Rhoi(X_i)<\infty$, $i\in[n]$, and $\eps>0$ arbitrary. Suppose $\mbf Z\in\prod_{i=1}^n\CS_i$ is such that for $\price_i(Z_i)\leq\Rhoi(X_i)+\frac{\eps}n$ and $X_i-Z_i\in\acc_i$, $i\in[n]$. Set $Z^*:=Z_1+...+Z_n$ and observe $X-Z^*\in\mathscr A_+$ as well as 
\[\sum_{i=1}^n\Rhoi(X_i)+\eps\geq \sum_{i=1}^n\price_i(Z_i)\geq\pi(Z^*)\geq\inf\{\pi(Z)\mid Z\in\mathcal M,~X-Z\in\mathscr A_+\}.\]
This proves \eqref{eq:help1}.
We now turn to \eqref{eq:help2}. If $\Lambda(X)=\infty$, assume for contradiction there is some $Z\in\mathcal M$ such that $X-Z\in \CL_0(\Lambda)\subset\sum_{i=1}^n\dom(\Rhoi).$
Choose $\mbf Y\in\alloc_{X-Z}$ such that $Y_i\in\dom(\Rhoi)$ for all $i$, and let $\mbf Z\in\alloc_Z^s$ be arbitrary. Then 
$$\Lambda(X)\leq\sum_{i=1}^n\Rhoi(Y_i+Z_i)=\sum_{i=1}^n\Rhoi(Y_i)+\sum_{i=1}^n\price_i(Z_i)<\infty.$$
This is a \textsc{contradiction} and no such $Z\in \mathcal M$ can exist. \eqref{eq:help2} holds in this case.
Now assume $X\in\dom(\Lambda)$ and suppose $Z\in \mathcal M$ satisfies $X-Z\in\CL_0(\Lambda)$. Hence, for arbitrary $\eps>0$ there is $\mbf Y\in\alloc_{X-Z}$ such that $\sum_{i=1}^n\Rhoi(Y_i)\leq \eps$. As $\mbf Y+\mbf Z\in\alloc_X$ for all $\mbf Z\in\alloc_Z^s$, 
\[\Lambda(X)\leq\inf_{\mbf Z\in\alloc_Z^s}\sum_{i=1}^n\Rhoi(Y_i+Z_i)=\sum_{i=1}^n\Rhoi(Y_i)+\pi(Z)\leq\eps+\pi(Z).\]
As $\eps>0$ was chosen arbitrarily, we obtain \eqref{eq:help2}.

(3) Assume (NSA) and (SUP) are fulfilled, let $\phi_0\in\CX^*$ as described in (SUP), and note that $\pi$ is linear by (1). We shall prove that $\phi_0|_{\CM}=\kappa\pi$ for some $\kappa>0$, so by rescaling $\phi_0|_\CM=\pi$ may be assumed without loss of generality. To this end, we restate requirement (SUP)(ii) as $\sup_{Y\in\mathscr A_+}\phi_0(Y)<\infty$, which entails positivity of the functional $\phi_0$ by monotonicity of $\mathscr A_+$. (SUP)(i) means in particular that $\phi_0|_{\ker(\pi)}\equiv 0$. For each $i\in[n]$ fix $U_i\in\CS_i\cap\CX_{++}$ such that $U:=\sum_{i=1}^nU_i$ satisfies $\pi(U)=\sum_{i=1}^n\price_i(U_i)=1$. As for all $Z\in\CM$, $Z-\pi(Z)U\in\ker(\pi)$, we infer $\phi_0(Z-\pi(Z)U)=0$, or equivalently $\phi_0=\phi_0(U)\pi$ on $\CM$. By the second part of (SUP)(i), $\phi_0(\tilde Z)\neq 0$ for some $\tilde Z\in\CM$ with $\pi(\tilde Z)\neq 0$. Using positivity of $\phi_0$, we obtain $0<\frac{\phi_0(\tilde Z)}{\pi(\tilde Z)}=\phi_0(U)$, hence  we may set $\kappa:=\phi_0(U)$. Finally, if $\kappa=1$, $X\in\CX$ is arbitrary, and $Z\in\CM$ is such that $X-Z\in\mathscr A_+$,
$$\pi(Z)=\phi_0(Z)=\phi_0(X)-\phi_0(X-Z)\geq \phi_0(X)-\sup_{Y\in\mathscr A_+}\phi_0(Y)>-\infty.$$
The bound on the right-hand side is independent of $Z$. Using the representation of $\Lambda$ in (2), properness follows.

(4)  Note that  $\Lambda$ is $\mathcal M$-additive by (2). Since $\Lambda$ is proper, we cannot have $\pi\equiv -\infty$, hence $\pi(0)=0$, i.e.\ (NSA) holds by (1). As regards positivity of $\pi$, choose $Y\in\CX$ with $\Lambda(Y)\in\R$. For $Z\in\CM\cap\CX_+$, monotonicity of $\Lambda$ then shows $\Lambda(Y)\leq \Lambda(Y+Z)=\Lambda(Y)+\pi(Z)$, which entails $\pi(Z)\geq 0$. It follows that $(\mathscr A_+,\CM,\pi)$ is a risk measurement regime.
\end{proof}

The preceding proposition offers a more geometric perspective on assumption (SUP). Suppose the agent system operates on a a locally convex Hausdorff topological Riesz space $(\CX,\preceq,\tau)$ and satisfies (NSA). Moreover, assume we can find a security $Z^*\in\CM$ such that
\begin{equation}\label{bad:security}Z^*\notin\cl_{\tau}(\mathscr A_++\ker(\pi)),\end{equation}
where here and in the the following $\cl_\tau(\cdot)$ denotes the closure of a set with respect to the topology $\tau$. Then (SUP) means $Z^*$ is a security which comes at a true cost for the market; it can be strictly separated from $\mathscr A_++\ker(\pi)$ using a linear functional $\phi_0\in\CX^*_+$, and this functional is exactly as described in (SUP).

In the situation of Proposition \ref{prop:operational}(4), the behaviour of the representative agent is given by the risk measurement regime $(\mathscr A_+,\CM,\pi)$. The risk sharing functional is the \textit{market capital requirement} associated to the \textit{market acceptance set} $\mathscr A_+$ and the global security market $(\CM,\pi)$.  

\subsection{Optimal payoffs and Pareto optima}

We now turn our attention to the existence of Pareto optimal allocations. By Proposition \ref{prop:equivalence}, $W\in\dom(\Lambda)$ admits a Pareto optimal allocation if and only if $\Lambda$ is \textit{exact} at $W$, i.e.\ \eqref{eq:exact} holds. We will see that this problem is closely related to the existence of  a market security $Z^W\in\CM$ which renders market acceptability $W-Z^W\in\mathscr A_+$ at the minimal price $\pi(Z^W)=\Lambda(W)$. The latter  problem has been studied by Baes et al. in \cite{Baes}.

\begin{definition}\label{def:optpayoff}
$W\in\CX$ admits an \textsc{optimal payoff} $Z^W\in\CM$ if $W-Z^W\in \mathscr A_+$ and $\pi(Z^W)=\Lambda(W)$.
\end{definition}

\begin{proposition}\label{prop:exact1}
Suppose that $\Lambda$ is proper. If $X\in\CX$ admits an optimal payoff $Z^X\in\CM$, then $X\in \dom(\Lambda)$ and $\Lambda$ is exact at $X$. In particular, for any $Y_i\in \acc_i$, $i\in [n]$, and $\mbf Z\in \alloc_{Z^X}^s$ such that $\sum_{i=1}^n Y_i=X-Z^X$ the allocation   $(Y_i+Z_i)_{i\in [n]}\in\alloc_X $ is Pareto optimal. If moreover $\CL_0(\Rhoi)=\acc_i+\ker(\price_i)$, $i\in[n]$, then $\Lambda$ is exact at $X\in \dom(\Lambda)$ if, and only if $X$ admits an optimal payoff. \end{proposition}

\begin{proof} As $\Lambda$ is proper, we have $\pi$ is linear, finite valued, and $\pi|_{\CS_i}=\price_i$, $i\in [n]$, by Proposition~\ref{prop:operational}. Assume $X\in\CX$ and $Z=Z^X\in\CM$ are such that $\Lambda(X)=\pi(Z)$ and $X-Z\in\mathscr A_+$. As $\pi(Z)\in \R$ and $\Lambda|_{\mathscr A_+}\leq 0$, $X\in \dom(\Lambda)$. Choose $Y_i\in\acc_i$, $i\in[n]$, such that $X-Z=\sum_{i=1}^nY_i$. For any $\mbf Z\in\alloc_Z^s$ we thus have $X=\sum_{i=1}^nY_i+Z_i$ and 
\[\Lambda(X)\leq \sum_{i=1}^n\Rhoi(Y_i+Z_i)=\sum_{i=1}^n\Rhoi(Y_i)+\sum_{i=1}^n\price_i(Z_i)\leq\pi(Z)=\Lambda(X),\]
where we have used $\Rhoi(Y_i)\leq 0$ and $\pi(Z)=\sum_{i=1}^n\price_i(Z_i)$  (Proposition~\ref{prop:operational}). This shows exactness of $\Lambda$ at $X$.

Now assume $\CL_0(\Rhoi)=\acc_i+\ker(\price_i)$, $i\in[n]$. Let $X\in \dom(\Lambda)$ and $\mbf X\in\alloc_X$ such that $\Lambda(X)=\sum_{i=1}^n\Rhoi(X_i)$. Further, let $U_i\in\CS_i$ with $\price_i(U_i)=1$. As $X_i-\Rhoi(X_i)U_i\in\acc_i+\ker(\price_i)$, $i\in[n]$, by assumption, we may find $\mbf N\in\prod_{i=1}^n\ker(\price_i)$ such that $X_i-\Rhoi(X_i)U_i+N_i\in\acc_i$ for every $i\in[n]$. The fact that $\sum_{i=1}^n\Rhoi(X_i)U_i-N_i$ is an optimal payoff for $X$ is immediately verified. 
\end{proof}

In a topological setting, the existence of optimal payoffs is intimately connected to the Minkowski sum $\mathscr A_++\ker(\pi)$ being closed:

\begin{proposition}\label{prop:exact2} Suppose $(\CX,\peq,\tau)$ is a topological Riesz space and $\Lambda$ is proper. Then $\mathscr A_++\ker(\pi)$ is closed if, and only if, $\Lambda$ is l.s.c.\ and every $X\in\dom(\Lambda)$ admits an optimal payoff.
\end{proposition}
\begin{proof} Suppose first that $\mathscr A_++\ker(\pi)$ is closed. For lower semicontinuity, we have to establish that $\mathcal L_c(\Lambda)$ is closed for every $c\in\R$. To this end, let $U_i\in\CS_i\cap \CX_{++}$ such that $\price_i(U_i)>0$ and set $U:=\sum_{i=1}^nU_i$. Without loss of generality, we may assume $\pi(U)=1$. We will show that 
\begin{equation}\label{eq:nullniveau}\mathcal L_c(\Lambda)=cU+\mathscr A_++\ker(\pi),\end{equation}
which is closed whenever $\mathscr A_++\ker(\pi)$ is closed.
The right-hand set in \eqref{eq:nullniveau} is included in the left-hand set by $\CM$-additivity of $\Lambda$. For the converse inclusion, let $X\in\mathcal L_c(\Lambda)$. For every $s>c$, there is a $Z_s\in\CM$ such that $c\leq \pi(Z_s)\leq s$ and $X-Z_s\in\mathscr A_+$. Consider the decomposition 
$$X-sU=X-Z_s+(\pi(Z_s)-s)U+Z_s-\pi(Z_s)U.$$
As $X-Z_s+(\pi(Z_s)-s)U\peq X-Z_s\in\mathscr A_+$ and $Z_s-\pi(Z_s)U\in\ker(\pi)$, monotonicity of $\mathscr A_+$ shows $X-sU\in\mathscr A_++\ker(\pi)$. Thus, $X-cU=\lim_{s\downarrow c}X-sU\in\cl_\tau(\augm)=\augm$, and \eqref{eq:nullniveau} is proved. Setting $c=0$ in \eqref{eq:nullniveau} shows $\mathcal L_0(\Lambda)=\augm$. Hence, $X-\Lambda(X)U\in\augm$ for all $X\in\dom(\Lambda)$, and for a suitable $N\in\ker(\pi)$ depending on $X$ we have $X-\Lambda(X)U+N\in\mathscr A_+$ and $\pi(\Lambda(X)U-N)=\Lambda(X)$. Therefore, an optimal payoff for $X$ is given by $\Lambda(X)U-N\in\CM$.

Assume now that $\Lambda$ is l.s.c.\ and that every $X\in\dom(\Lambda)$ allows for an optimal payoff. Let $(X_i)_{i\in I}$ be a net in $\mathscr A_++\ker(\pi)$ converging to $X\in\CX$. Then $\Lambda(X)\leq 0$ by lower semicontinuity of $\Lambda$. Let $Z\in\mathcal M$ be an optimal payoff for $X$, so that $\pi(Z)=\Lambda(X)\leq 0$. For $U$ as above and $Y:=X-Z\in\mathscr A_+$ we obtain $Y+\pi(Z)U\in\mathscr A_+$ by monotonicity of $\mathscr A_+$. Also $Z-\pi(Z)U\in\ker(\pi)$. Thus $X=\left(Y+\pi(Z)U\right)+\left(Z-\pi(Z)U\right)\in\mathscr A_++\ker(\pi).$
\end{proof}

Proposition~\ref{prop:exact2} is related to \cite[Proposition 4.1]{Baes}. Together with Proposition \ref{prop:exact1}, it is a powerful sufficient condition for the existence of Pareto optima which we shall apply in Sections \ref{sec:polyagent} and \ref{sec:lawinvacc}. The only non-trivial steps will be to verify the properness of $\Lambda$ and closedness of $\augm$.

\subsection{Existence of equilibria}

We proceed with the discussion of equilibria in the very general case when market losses are modelled by a \textit{Fr\'echet lattice} $(\CX,\peq,\tau)$. As this notion is ambiguous in the literature, we emphasise that a Fr\'echet lattice is a locally convex-solid topological Riesz space whose topology is completely metrisable. 

In particular, \textit{Banach lattices} are Fr\'echet lattices. As a more general example, one may consider the Wiener space $C([0,\infty))$ of all continuous functions on the non-negative half-line with the pointwise oder $\leq$ and the topology $\tau_D$ arising from the metric
\[D(f,g):=\sum_{k=1}^\infty2^{-k}\frac{\max_{0\leq r\leq k}|f(r)-g(r)|}{1+\max_{0\leq r\leq k}|f(r)-g(r)|},\quad f,g\in C([0,\infty)).\]
Clearly, $(C([0,\infty)),\leq,\tau_D)$ is not a Banach lattice, but a Fr\'echet lattice. Its choice as model space is justified if the primitives in question are continuous trajectories of, e.g., the net value of some good over time. 

For the following main meta theorem proving the existence of equilibria in case the model space is a Fr\'echet lattice, recall the definition of the jointly accepted securities, $\check{\CS}:=\bigcap_{i=1}^n\CS_i$. Moreover, we set here and in the following $\tn{int}\,\dom(\Lambda)$ to be the $\tau$-interior of the effective domain of the risk sharing functional $\Lambda$. Given a proper function $f:\CX\to(-\infty,\infty]$, its \textsc{dual conjugate} is the function $f^*:\CX^*\to(-\infty,\infty]$ defined by $f^*(\phi)=\sup_{X\in\CX}\phi(X)-f(X)$. Given $X\in\dom(f)$, $\phi\in\CX^*$ is a \textsc{subgradient} of $f$ at $X$ if $f(X)=\phi(X)-f^*(\phi)$. 

\begin{proposition}\label{prop:existenceequilibria} 
Suppose $\CX$ is a Fr\'echet lattice and that $\Lambda$ is l.s.c. and proper. Moreover, let \tn{(NR)} be satisfied, i.e.\ there is a $\tilde Z\in\check{\CS}$ with $\pi(\tilde Z)\neq 0$. If $\mbf W\in \prod_{i=1}^n\CX_i$ is such that $W:=W_1+...+W_n\in\tn{int}\,\dom(\Lambda)$ and there exists a Pareto optimal allocation of $W$, there is an equilibrium $(\mbf X,\phi)$ of $\mbf W$. 
\end{proposition}
\begin{proof}
Fix $\mbf W\in\prod_{i=1}^nW_i$ such that $W:=W_1+...+W_n\in\tn{int}\,\dom(\Lambda)$. As a Fr\'echet lattice is a barrelled space, $\Lambda$ is subdifferentiable at $W$ by \cite[Corollary 2.5 \& Proposition 5.2]{eketem}, i.e. there is a subgradient $\phi\in\CX^*$ of $\Lambda$ at $W$ satisfying $\Lambda(W)=\phi(W)-\Lambda^*(\phi)$. As $\Lambda$ is monotone, $\phi\in\CX_+^*$ , and by Lemma~\ref{lem:sumdual}
\begin{equation}\label{eq:dualSigma}\Lambda^*(\phi)=\sum_{i=1}^n \Rhoi^*(\phi|_{\CX_i}), \quad \phi\in\CX^*.\end{equation}
Let $\mbf Y$ be any Pareto optimal allocation of $W$. As $\Lambda(W)=\sum_{i=1}^n\Rhoi(Y_i)\in \R$, $\Lambda(W)$, $\Lambda^\ast(\phi)$ and $\Rhoi^\ast(\phi|_{\CX_i})$, $i\in [n]$, are all real numbers. Also, as 
$$\infty>\Rhoi^*(\phi|_{\CX_i})\geq \sup_{Z\in\CS_i}\phi(Y_i+Z)-\Rhoi(Y_i+Z)=\phi(Y_i)-\Rhoi(Y_i)+\sup_{Z\in\CS_i}\phi(Z)-\price_i(Z) ,$$
$\phi|_{\CS_i}=\price_i$, $i\in [n]$, has to hold, which in turn implies $\phi|_{\CM}=\pi$ by linearity of $\pi$ and Proposition~\ref{prop:operational}.
By (NR), we may fix $\tilde Z\in\check{\CS}$ such that $\pi(\tilde Z)=1=\price_i(\tilde Z), i\in [n]$. Let 
\[X_i:=Y_i+\phi(W_i-Y_i)\tilde Z,\quad i\in[n].\]
Note that $\mbf X\in\alloc_W$ holds because $\sum_{i=1}^nW_i=\sum_{i=1}^nY_i=W$ and thus $\sum_{i=1}^nX_i=W$. Moreover, $\mbf X$ is Pareto optimal: 
\[\sum_{i=1}^n\Rhoi(X_i)=\sum_{i=1}^n\Rhoi(Y_i)+\phi(W_i-Y_i)\pi(\tilde Z)=\sum_{i=1}^n\Rhoi(Y_i)+\phi(W-W)= \sum_{i=1}^n\Rhoi(Y_i)=\Lambda(W).\]
Also, as $\phi(X_i)-\Rhoi^*(\phi|_{\CX_i})\leq \Rhoi(X_i)$ for all $i\in [n]$ and
\[\sum_{i=1}^n\Rhoi(X_i)=\Lambda(W)=\phi(W)-\Lambda^*(W)=\sum_{i=1}^n\phi(X_i)-\Rhoi^*(\phi|_{\CX_i}),\]
$\Rhoi(X_i)=\phi(X_i)-\Rhoi^*(\phi|_{\CX_i})$ has to hold for all $i\in [n]$.
We claim that $(\mbf X,\phi)$ is an equilibrium. Indeed, as $\phi(-X_i)=\phi(-W_i)$ holds for all $i\in[n]$, the budget constraints are satisfied. Moreover, if $i\in[n]$ and $Y\in\CX_i$ satisfies $\phi(-Y)\leq \phi(-W_i)=\phi(-X_i)$, we obtain
$$\Rhoi(Y)\geq \phi(Y)-\Rhoi^*(\phi|_{\CX_i})\geq \phi(X_i)-\Rhoi^*(\phi|_{\CX_i})=\Rhoi(X_i).$$

\end{proof}


\section{Polyhedral agent systems}\label{sec:polyagent}

In this section we assume that the agent system $(\mathcal R_1,...,\mathcal R_n)$ operates on a market space $\CX$ given by a Fr\'echet lattice. Each agent $i\in[n]$ operates on a closed ideal $\CX_i\subset\CX$, and $\CX_1+...+\CX_n=\CX$. The assumption of closedness implies that $(\CX_i,\peq,\tau\cap\CX_i)$ is a Fr\'echet lattice in its own right. We will assume that each acceptance set $\acc_i\subset\CX_i$ is \textit{polyhedral}.

\begin{definition}\label{def:poly}
Let $(\CX,\peq,\tau)$ be a Fr\'echet lattice. A convex set $\CC\subset\CX$ is called \textsc{polyhedral} if there is a finite set $\mathcal J\subset\CX^*$ and $\beta\in\R^{\mathcal J}$ such that 
\begin{center}$\acc=\{X\in\CX\mid \forall\,\phi\in\mathcal J:~\phi(X)\leq\beta_\phi\}.$\end{center}
An agent system $(\mathcal R_1,...,\mathcal R_n)$ is \textsc{polyhedral} if it has properties (NSA) and (SUP), and  each acceptance set $\acc_i$, $i\in[n]$, is polyhedral.
\end{definition}

Polyhedrality of a set $\CC$ is equivalent to the existence of some $m\in\N$, a continuous linear operator $T:\CX\to\R^m$, and $\beta\in\R^m$ such that $\CC=\{X\in\CX\mid T(X)\leq\beta\}$, where the defining inequality is understood coordinatewise. In case of an acceptance set, the representing linear operator can be chosen to be positive. Risk measures with polyhedral acceptance sets play a prominent role in Baes et al. \cite{Baes}, where the set of optimal payoffs for a single such risk measure is studied.  

\begin{example}\label{ex:poly1}
For the sake of simplicity, we consider a finite-dimensional setting. Let $A$, $B$ and $C$ be three finite and disjoint sets of scenarios for the future state of the economy, either suggested by the internal risk management or a regulatory authority. We set $\Omega_1=A\cup B$, $\Omega_2:=B\cup C$, which are the scenarios relevant for agent $i\in\{1,2\}$. $B$ can be seen as a non-trivial set of jointly relevant scenarios, and $\Omega:=A\cup B\cup C$ is the set of scenarios that are relevant to the whole system. However, as the relevant scenarios for agent $i$ are $\omega\in\Omega_i$, it is both individually and systemically rational of her to demand that her stake in the sharing of a market loss is \textit{neutral} in scenarios $\omega\in\Omega\backslash\Omega_i$. The canonical choice of the model spaces is in consequence 
$$\CX:=\{X:\Omega\to\R\},\quad \CX_i:=\{X\in\CX\mid X|_{\Omega\backslash\Omega_i}\equiv 0\},~i=1,2.$$
For illustration, we assume individual acceptability is defined in terms of scenariowise loss constraints: let $K_1\in\CX_1$ and $K_2\in\CX_2$ be two arbitrary, but fixed vectors of individual loss tolerances. Consider the risk measurement regimes 
\begin{center}$\acc_1:=\{X\in\CX_1\mid X\leq K_1\},\quad \CS_1=\tn{span}\{\ind_A,\ind_B\},\quad\price_1(x\ind_A+y\ind_B)=x+y$,\\
$\acc_2:=\{X\in\CX_2\mid X\leq K_2\},\quad\CS_2=\tn{span}\{\ind_B,\ind_C\},\quad \price_2(x\ind_B+y\ind_C)=x+y.$\end{center}
The (Arrow-Debreu type) securities $\ind_A$, $\ind_B$ and $\ind_C$, respectively, pay off a unit amount in case one of the scenarios of $A$, $B$, and $C$, respectively, realises. The objective is not to exceed the loss tolerances $K_1$ and $K_2$ at minimal cost. 
\end{example}

\subsection{Existence of optimal payoffs, Pareto optima and equilibria}

We turn to the existence of optimal allocations in the setting introduced above. By definition, a polyhedral agent system satisfies (NSA) and (SUP). The resulting risk sharing functional $\Lambda$ is proper by Proposition \ref{prop:operational}(3). By Propositions \ref{prop:exact1} and \ref{prop:exact2}, the existence of Pareto optimal allocations would be proved if closedness of $\augm$ can be established. 

For the following lemma, recall that a Fr\'echet space is a completely metrisable locally convex topological vector space. In particular, every Fr\'echet lattice is a Fr\'echet space. 
\begin{lemma}\label{lem:linpoly}Let $\CX$ be a Fr\'echet space.
\begin{enumerate}
\item[\tn{(1)}] A subset $\CC\subset \CX$ is a polyhedron if, and only if, there are closed subspaces $\CX^1,\CX^2\subset\CX$ such that $\CX=\CX^1\oplus\CX^2$, dim$(\CX^2)<\infty$, and $\CC=\CX^1+\CC'$ for a polyhedron $\CC'\subset\CX^2$.
\item[\tn{(2)}] Suppose $\CY$ and $\CX$ are Fr\'echet spaces, $\CC\subset\CY$ is polyhedral, and $T:\CY\to\CX$ is a surjective linear operator. Then $T(\CC)$ is polyhedral in $\CX$. 
\end{enumerate}
\end{lemma}
\begin{proof} 
(1) Combine the proof of \cite[Corollary 2.1]{Zheng} with the Closed Graph Theorem \cite[Theorem 5]{Husain}.

(2) By (1), there are two closed subspaces $\CY^1,\CY^2\subset\CY$ such that $\CY=\CY^1\oplus\CY^2$, dim$(\CY^2)<\infty$, and $\CC=\CY^1+\CC'$ for a polyhedron $\CC'$ in the finite-dimensional subspace $\CY^2$. Define $\CX^2:=T(\CY^2)$, which is finite-dimensional. Every finite-dimensional subspace of a Fr\'echet space is complemented by a closed subspace. Thus $\CX=\CX^1\oplus\CX^2$ for a closed subspace $\CX^1$. Clearly, $T(\CC')\subset\CX^2$ is a polyhedron. Moreover, denoting by $\gamma_i:\CX\to\CX^i$ the projection in $\CX$ onto the linear subspaces $\CX^i$, surjectivity of $T$ implies $\CX^1=\gamma_1(\CX)=\gamma_1(T(\CY^1))+\gamma_1(T(\CY^2)=\gamma_1(T(\CY^1))$. Moreover, 
$$T(\CC)=T(\CY^1)+T(\CC')=\CX^1+\gamma_2(T(\CY^1))+T(\CC').$$
$\gamma_2(T(\CY^1))$ is as subspace of the finite-dimensional space $\CX^2$ a polyhedron, and so is the sum $\gamma_2(T(\CY^1))+T(\CC')$ of two finite-dimensional polyhedra. Conclude with (1). \end{proof}

\begin{theorem}\label{A+ker poly}
Let $(\mathcal R_1,...,\mathcal R_n)$ be a polyhedral agent system on a Fr\'echet lattice $\CX$. Then the set $\augm$ is proper, polyhedral, and closed, $\Lambda$ is ls.c., and every $X\in\dom(\Lambda)$ admits an optimal payoff $Z^X\in\CM$, and can thus be allocated Pareto optimally as in Proposition~\ref{prop:exact1}. 
\end{theorem}
\begin{proof}
The set $\augm$ is proper by assumption (SUP). Moreover, it is polyhedral: consider the Fr\'echet space $\CY:=(\prod_{i=1}^n\CX_i)\times \ker(\pi)$.\footnote{ $\CY$ is not a Fr\'echet lattice, hence the necessity for the above formulation of Lemma \ref{lem:linpoly}.} By assumption, the set $\CC:=(\prod_{i=1}^n\acc_i)\times\ker(\pi)$ is polyhedral, and $T:\CY\to\CX$ defined by $T(X_1,...,X_n,N)=\sum_{i=1}^nX_i+N$ is surjective and linear. As $\CX$ is a Fr\'echet space, Lemma \ref{lem:linpoly}(2) yields the polyhedrality of $T(\CC)=\mathscr A_++\ker(\pi)$. As a polyhedron, it is automatically closed. Since $\Lambda$ is proper, it is l.s.c. and optimal payoffs exist for every $X\in\dom(\Lambda)$ by Proposition~\ref{prop:exact2}.\end{proof}

Theorem~\ref{A+ker poly} in conjunction with Proposition~\ref{prop:existenceequilibria} imply the existence of equilibria: 

\begin{corollary}\label{cor:polyequilibrium} 
If a polyhedral agent system $(\mathcal R_1,...,\mathcal R_n)$ on a Fr\'echet lattice $\CX$ satisfies \tn{(NR)}, for every $\mbf W\in\prod_{i=1}^n\CX_i$ such that $W_1+...+W_n\in\tn{int}\,\dom(\Lambda)$ there is an equilibrium $(\mbf X,\phi)$. 
\end{corollary}

\subsection{Lower hemicontinuity of the Pareto optima correspondence}

In this section we consider the correspondence $\mathcal P$ mapping $X\in\dom(\Lambda)$ to its Pareto optimal allocations $\mbf X\in\alloc_X$. Invoking Proposition \ref{prop:equivalence}, we can represent 
\begin{equation}\label{eq:defPcal}\mathcal P(X)=\left\{\mbf X\in\alloc_X\Big| \Lambda(X)=\sum_{i=1}^n\Rhoi(X_i)\right\}.\end{equation}
For a brief summary of terminology concerning and properties of correspondences (or set-valued maps), we refer to Appendix \ref{appendixC}. The following theorem asserts that $\CP$ is lower hemicontinuous under mild conditions.

\begin{theorem}\label{thm:paretopoly}
Assume that for a polyhedral agent system the market space $\CX$ is finite-dimensional or $\CX_i=\CX$ for all $i\in[n]$. Then the correspondence $\mathcal P$ is lower hemicontinuous on $\dom(\Lambda)$, and admits a continuous selection on $\dom(\Lambda)$.
\end{theorem}

Its proof requires the following highly technical Lemmas~\ref{lhc Gamma2:fin} and~\ref{lhc Gamma2:infin} whose proofs imitate in parts a technique from Baes et al. \cite{Baes}. Note that in analogy with Theorem \ref{A+ker poly}, $\mathscr A_+$ is closed.

\begin{lemma}\label{lhc Gamma2:fin}If $\acc_i\subset \CX_i$, $i\in[n]$, are polyhedral acceptance sets and $\CX$ is finite-dimensional, the correspondence $\Gamma:\mathscr A_+\ni X\to\alloc_X\cap \prod_{i=1}^n\acc_i$ is lower hemicontinuous. 
\end{lemma}
\begin{proof} 
Each subspace $\CX_i$ is polyhedral, as well. As in Definition \ref{def:poly}, for each $i\in[n]$ we fix $m_i\in\N$, a positive linear and continuous operator $T_i:\CX\to\R^{m_i}$, and vectors $\beta_i\in\R^{m_i}$, such that
\[\acc_i=\{X\in\CX\mid T_i(X)\leq\beta_i\}.\]

\textit{Step 1:} For fixed $X\in\mathscr A_+$ we decompose $\Gamma(X)$ as the sum of a universal and an $X$-dependent component. Recall from Appendix \ref{appendix:geom} that the recession cone of $\Gamma(X)$ is given by
$0^+\Gamma(X):=\{\mbf Y\mid \forall\, \mbf X\in\Gamma(X)\,\forall\, k>0: \mbf X+k\mbf Y\in\Gamma(X)\}$. The lineality space of $\Gamma(X)$ is $0^+\Gamma(X)\cap(-0^+\Gamma(X))=\{\mbf Y\in\alloc_0\mid\forall\,i\in[n]:~T_i(Y_i)=0\},$
a subspace independent of $X$. By virtue of Lemma \ref{lem:decomp}, there is a $X$-independent subspace $\mathcal V\subset\prod_{i=1}^n\CX_i$ such that
$$\Gamma(X)=\alpha(X)+0^+\Gamma(X),\quad\alpha(X):=\tn{co}(\tn{ext}(\Gamma(X)\cap \mathcal V)),$$
where co$(\cdot)$ denotes the convex hull operator and ext$(\Gamma(X)\cap \mathcal V)$ the set of extreme points of $\Gamma(X)\cap \mathcal V$. 

\textit{Step 2:} In this step, we prove that the correspondence $\alpha:\mathscr A_+\to\prod_{i=1}^n\acc_i$ maps bounded sets to bounded sets. To this end, let $D:=\dim(\CX)=\dim(\CX^*)$ and choose a basis $\psi_1,...,\psi_{D}$ of $\CX^*$. Note that $\mbf X\in\Gamma(X)\cap \mathcal V$ if, and only if, 
\begin{itemize}
\item $\mbf X$ is an allocation of $X$, i.e. $\psi_j(X_1+...+X_n)=\psi_j(X)$ for all $j\in [D]$, or equivalently $\psi_j(X_1+...+X_n)\leq \psi_j(X)$ and $(-\psi_j)(X_1+...+X_n)\leq(-\psi_j)(X)$;
\item each $X_i$ lies in $\acc_i$, i.e. $T_i(X_i)\leq \beta_i$; 
\item $\mbf X\in \mathcal V$.
\end{itemize}
Clearly, the properties listed above describe a polyhedral set; more precisely, for $m:=\sum_{i=1}^nm_i+2D$, $m_i$ defined above, we may find a continuous linear operator $\mathbf{S}: \mathcal V\to\R^m$ and a continuous function $\mathbf{f}:\CX\to\R^m$ such that 
$$\Gamma(X)\cap \mathcal V=\{\mbf X\in \mathcal V\mid \mathbf{S}(\mbf X)\leq \mathbf{f}(X)\}.$$
Every ``row'' $\mathbf{S}_i$ of $\mathbf{S}$ corresponds to an element of $\mathcal V^*$. By \cite[Theorem II.4.2]{Barvinok}, for every extreme point $\mbf X\in \alpha(X)\cap \mathcal V$ the set $I(\mbf X)=\{i\in[m]\mid \mathbf{S}_i(\mbf X)=\mathbf{f}_i(X)\}$, which contains at least $\dim(\mathcal V)$ elements, satisfies that $\tn{span}\{\mathbf{S}_i\mid i\in I(\mbf X)\}=\mathcal V^*.$
Let $\mathbb F(X):=\{I(\mbf X)\mid \mbf X\in\tn{ext}(\Gamma(X)\cap \mathcal V)\}$ be the collection of all such $I(\mbf X)$ corresponding to an extreme point. Its cardinality is bounded by the finite number of extreme points, the latter depending on $\dim(\mathcal V)$ and $m$ only. Moreover, for each $I\in\mathbb F(X)$, the linear operator
$\mathbf{S}_I:\mathcal V\ni\mbf Y\mapsto (\mathbf{S}_i(\mbf Y))_{i\in I}$
is injective and thus invertible on its image. Keeping  \cite[Corollary II.4.3]{Barvinok} in mind, we have shown $(\mathbf{S}_I)_{I\in \mathbb F(X)}$ is a finite family of invertible operators whose cardinality depends on $\dim(\mathcal V)$ and $m$ only.
Let $\CB\subset\CX$ be a bounded set. For each $I\in\mathbb F(X)$, $\mathbf{f}_I$ is continuous and thus maps $\CB$ to a bounded set. Also, $\mathbf S_I^{-1}$ is continuous by the Closed Graph Theorem \cite[Theorem 5]{Husain}, whence boundedness of $\{\mathbf S_I^{-1}(\mathbf f_I(X))\mid X\in A\}$ follows. Recall that $\{\mathbb F(X)\mid X\in \CB\}$ is finite. Using Carath\'eodory's Theorem \cite[Theorem 17.1]{Rockafellar}, co$\{\mathbf S_I^{-1}(\mathbf f_I(X))\mid X\in \CB, I\in\mathbb F(X)\}$ is bounded. As 
\[\bigcup_{X\in \CB}\alpha(X)=\bigcup_{X\in \CB}\tn{co}\{\mathbf{S}_I^{-1}(\mathbf{f}_I(X))\mid I\in \mathbb F(X)\}\subset \tn{co}\{\mathbf S_I^{-1}(\mathbf f_I(X))\mid X\in \CB,\,I\in\mathbb F(X)\},\]
it has to be bounded as well and Step 2 is proved. 

\textit{Step 3:} $\Gamma$ is lower hemicontinuous. Let $(X^k)_{k\in\N}\subset\mathscr A_+$ be convergent to $X\in\mathscr A_+$ and let $\mbf X\in\Gamma(X)$. We have to show that there is a subsequence $(k_\lambda)_{\lambda\in\N}$ and $\mbf X^\lambda\in\Gamma(X^{k_\lambda})$ such that $\mbf X^\lambda\to\mbf X$; c.f. Appendix \ref{appendixC}. To this end, let first $\mbf Y^k\in\alpha(X^k)$, $k\in\N$, which is a bounded sequence by Step 2. After passing to a subsequence $(k_\lambda)_{\lambda\in\N}$, we may assume $\mbf Y^{k_\lambda}\to \mbf Y\in\Gamma(X)$ (as $\acc_i$ is closed, $i\in[n]$). If $\Gamma(X)$ is a singleton, $\mbf Y=\mbf X$ has to hold and we may choose $\mbf X^\lambda:=\mbf Y^{k_\lambda}$. Otherwise, suppose first that $\mbf X$ lies in the \textit{relative interior} of $\Gamma(X)$, i.e. there is an $\eps>0$ such that $\mbf X+\eps(\mbf X-\mbf Y)\in \Gamma(X)$, as well. Recall the definition of the linear operators $T_i$, $i\in[n]$ above and fix $i\in[n]$. Let $1\leq j\leq m_i$ be arbitrary. We denote by $T_i^j(W)$ the $j$-th entry of $T_i(W)$.\\
\underline{Case 1:} $T_i^j(X_i)=\beta_j$. From $Y_i\in\acc_i$, we infer 
$$0\geq T^j_i(X_i+\eps(X_i-Y_i))-\beta_j=\eps(\beta_j-T^j_i(Y_i))\geq 0,$$
which means $T^j_i(Y_i)=\beta_j$, as well. Set $\lambda(i,j)=1$. \\
\underline{Case 2:} $T^j_i(X_i)<\beta_j$. As $Y^{k_\lambda}_i\to Y_i$ for $\lambda\to\infty$, there must be a $\lambda(i,j)\in\N$ such that for all $\lambda\geq \lambda(i,j)$ 
$$T_i^j(Y^{k_\lambda}_i-Y_i+X_i)\leq \beta_j.$$
Hence for all $\lambda\geq \max_{i\in[n], 1\leq j\leq m_i}\lambda(i,j)$, one obtains 
$$\mbf X^{\lambda}:=\mbf Y^{k_\lambda}-\mbf Y+\mbf X\in\prod_{i=1}^n\acc_i\cap\alloc_{X^{k_\lambda}}=\Gamma(X^{k_\lambda}),$$
and $\mbf X^{\lambda}\to\mbf X$. It remains to notice that each $\mbf X\in\Gamma(X)$ may be approximated with a sequence in the relative interior of $\Gamma(X)$, c.f. \cite[Theorem 6.3]{Rockafellar}. The assertion is proved. \end{proof}

\begin{lemma}\label{lhc Gamma2:infin}If $\acc_i\subset \CX_i$, $i\in[n]$, are polyhedral acceptance sets and $\CX_i=\CX$ for all $i\in[n]$, the correspondence $\Gamma:\mathscr A_+\ni X\to\alloc_X\cap \prod_{i=1}^n\acc_i$ is lower hemicontinuous. 
\end{lemma}
\begin{proof}
We shall derive the assertion from Lemma \ref{lhc Gamma2:fin}. For $i\in[n]$ fixed, Lemma \ref{lem:linpoly}(1) allows to find closed subspaces $\CX_i^1,\CX_i^2\subset\CX$, the latter of finite dimensions, and polyhedral acceptance sets $\tilde\acc_i\subset\CX_i^2$ such that $\CX=\CX_i^1\oplus\CX_i^2$ and $\acc_i=\CX_i^1+\tilde\acc_i$. The space $\CX^2:=\sum_{i=1}^n\CX_i^2$ is finite-dimensional and thus complemented in $\CX$ by a closed subspace $\CX^1$, and the projections $\gamma_i:\CX\to\CX^i$ are continuous linear operators. Hence, we may rewrite $\acc_i=\CX^1+\mathcal B_i$ and $\mathscr A_+=\CX^1+\mathscr B_+$, where $\mathcal B_i=\tilde \acc_i+\gamma_2(\CX^1_i)$ and $\mathscr B_+:=\sum_{i=1}^n\mathcal B_i$. Furthermore, $\mathscr A_+=\{X\in\CX\mid\gamma_2(X)\in\mathscr B_+\}$
and $\tilde\Gamma:\mathscr B_+\ni Y\to\prod_{i=1}^n\mathcal B_i\cap\{\mbf Y\in (\CX^2)^n\mid \sum_{i=1}^nY_i=Y\}$ is lower hemicontinuous by Lemma \ref{lhc Gamma2:fin}.\\
Let now $(X^k)_{k\in\N}\subset\mathscr A_+$ be a sequence converging to $X\in\mathscr A_+$, which implies $\gamma_2(X^k)\to\gamma_2(X)$ for $k\to\infty$. Let $\mbf X\in \Gamma(X)$, whence $\mbf Y:=(\gamma_2(X_1),...,\gamma_2(X_n))\in\tilde\Gamma(\gamma_2(X))$ follows. As $\tilde\Gamma$ is lower hemicontinuous, there is a subsequence $(k_\lambda)_{\lambda\in\N}$ and $\mbf Y^\lambda\in \tilde\Gamma(\gamma_2(X^{k_\lambda}))$ such that $\mbf Y^{\lambda}\to \mbf Y$. Now we define $\mbf X^\lambda\in\CX^n$ by setting $X_i^\lambda:=\gamma_1(X_i)+\frac 1 n\gamma_1(X^{k_\lambda}-X)$ and note that $\sum_{i=1}^nX^\lambda_i=\gamma_1(X^{k_\lambda})$. Lower hemicontinuity of $\Gamma$ follows once we observe $\mbf X^\lambda+\mbf Y^\lambda\in\Gamma(X^{k_\lambda})$ and $\mbf X^\lambda+\mbf Y^\lambda\to\mbf X$. 
\end{proof}

\begin{proof}[Proof of Theorem \ref{thm:paretopoly}]
In addition to the correspondence $\mathcal P$ defined by \eqref{eq:defPcal} consider the following three correspondences:
\begin{itemize}
\item $\Gamma_1:\dom(\Lambda)\two \CM,~X\mapsto\{Z\in\CM\mid X-Z\in\mathscr A_+,~\Lambda(X)=\pi(Z)\},$ which is lower hemicontinuous on $\dom(\Lambda)$ by virtue of the polyhedrality of $\mathscr A_+$ and \cite[Theorem 5.11]{Baes}.
\item $\Gamma_2:\mathscr A_+\two\prod_{i=1}^n\acc_i,~X\mapsto\alloc_X\cap\prod_{i=1}^n\acc_i,$ which is lower hemicontinuous by Lemmas \ref{lhc Gamma2:fin} and \ref{lhc Gamma2:infin}, respectively.
\item $\Gamma_3:\CM\two\prod_{i=1}^n\CS_i,~Z\mapsto\alloc_Z^s,$ which is lower hemicontinuous by Lemma \ref{lem:contsel}.
\end{itemize}
Applying \cite[Theorem 17.23]{Ali}, $\Gamma: \dom(\Lambda)\ni X\mapsto \Gamma_2(\{X\}-\Gamma_1(X))+\Gamma_3(\Gamma_1(X))$ is lower hemicontinuous as well. \\
In fact, $\Gamma=\CP$ holds. To see this, let $X\in\dom(\Lambda)$ be arbitrary. $\Gamma_2(\{X\}-\Gamma_1(X))+\Gamma_3(\Gamma_1(X))\subset \CP(X)$ follows from the proof of Proposition~\ref{prop:exact1}. For the converse inclusion, let $\mbf X\in\CP(X)$ be arbitrary. Choose $Z_i\in\CS_i$, $i\in[n]$, such that $X_i-Z_i\in\acc_i$ and $\Rhoi(X_i)=\price_i(Z_i)$, which is possible by Theorem \ref{A+ker poly} in the case $n=1$. Let $Z=Z_1+...+Z_n$ and note that
\begin{center}$\pi(Z)=\sum_{i=1}^n\price_i(Z_i)=\sum_{i=1}^n\Rhoi(X_i)=\Lambda(X),$\end{center}
i.e. $Z\in\Gamma_1(X)$. Moreover, as $\mbf X-\mbf Z\in \Gamma_2(X-Z)\subset\Gamma_2(\{X\}-\Gamma_1(X))$, it only remains to note $\mbf X=(\mbf X-\mbf Z)+\mbf Z\in\Gamma_2(\{X\}-\Gamma_1(X))+\Gamma_3(\Gamma_1(X))$. Equality of sets is established. 

Finally, $\dom(\Lambda)$ is metrisable and therefore paracompact; c.f. \cite{Rudin}. Moreover, $\CX^n$ is a Fr\'echet space, and as $\CP:\dom(\Lambda)\two\prod_{i=1}^n\CX_i$ has non-empty closed convex values, a continuous selection for $\mathcal P$ exists by the Michael Selection Theorem \cite[Theorem 17.66]{Ali}. 
\end{proof}

One may wonder whether the correspondence $\mathcal E:\prod_{i=1}^n\CX_i\two\prod_{i=1}^n\CX_i\times\CX^*$
mapping an initial loss endowment $\mbf W$ to all its equilibrium allocations such that $\mbf X\in\mathcal P(W_1+...+W_n)$ and $\phi$ is a subgradient of $\Lambda$ at $W_1+...+W_n$, as in the proof of Proposition~\ref{prop:existenceequilibria} is lower hemicontinuous under suitable conditions. This, however, is not the case. Suppose $\CX$ admits two positive functionals $\phi,\psi\in\CX^*_+$ such that $\ker(\phi)\backslash\ker(\psi)\neq\emptyset$. We assume $n=1$ and consider an agent system $\mathcal R=(\acc,\CS,\price)$ such that $\Rho(X)=\max\{\phi(X),\psi(X)\}$, $X\in\CX$. Let $W\in\CX$ such that $\phi(W)=0<\psi(W)$. Thus, for all $n\in\N$, the equilibrium price at $\frac 1 nW$ would be $\psi$, whereas any element of $co(\{\phi,\psi\})$ could be chosen as equilibrium price at 0. $\mathcal E$ is \textit{not} lower hemicontinuous in this case. 

\subsection{An example}

We close this section by showing how Pareto optima can be computed in the situation of Example~\ref{ex:poly1}. Note that for $x,y\in\R$, we have 
$$X-x\ind_A-y\ind_B\in\acc_1~\iff~\max_{a\in A}X(a)-K_1(a)\leq x\tn{ and }\max_{b\in B}X(b)-K_1(b)\leq y.$$
Consequently,
$$\rho_1(X):=\rho_{\mathcal R_1}(X)=\max_{a\in A}X(a)-K_1(a)+\max_{b\in B}X(b)-K_1(b),\quad X\in\CX_1,$$
and it only takes finite values. An analogous computation shows 
$$\rho_2(X):=\rho_{\mathcal R_2}(X)=\max_{b\in B}X(b)-K_2(b)+\max_{c\in C}X(c)-K_2(c),\quad X\in\CX_2.$$
which also takes only finite values. One easily proves that $(\mathcal R_1,\mathcal R_2)$ is a polyhedral agent system and that the representative agent for it is given by 
\begin{center}$\mathscr A_+=\acc_1+\acc_2=\{X\in\CX\mid X\leq \tilde K:=K_1+K_2\},~\CM=\tn{span}\{\ind_A,\ind_B,\ind_C\}$,\\
$\pi(x\ind_A+y\ind_B+z\ind_C)=x+y+z,\quad x,y,z\in\R.$\end{center}
Furthermore
\[\ker(\pi)=\{N_{x,y}:=x\ind_A-(x+y)\ind_B+y\ind_C\mid x,y\in\R\}.\]
We now aim to compute the associated risk sharing functional $\Lambda$ and Pareto optimal allocations. To this end, for $X\in\CX$, we introduce the notation $\rho^A(X):=\max_{a\in A} X(a)-K_1(a)$, $\rho^B(X):=\max_{b\in B}X(b)-\tilde K(b)$, and $\rho^C(X):=\max_{c\in C} X(c)-K_2(c)$. Using the characterisation of $\mathscr A_+$, one obtains 
\[\augm=\{X\in\CX\mid \rho^B(X)\leq -\rho^A(X)-\rho^C(X)\}.\]
A straightforward computation yields 
\[\Lambda(X)=\inf\{r\in\R\mid X-r\ind_B\in \augm\}=\rho^A(X)+\rho^B(X)+\rho^C(X).\]
Note that $X-\Lambda(X)\ind_B-N_{\rho^A(X),\,\rho^C(X)}\in\mathscr A_+$, since
\[\left((X-\rho^A(X))\ind_A+K_1\ind_B,(X-\rho^B(X)-K_1)\ind_B+(X-\rho^C(X))\ind_C\right)\]
is an allocation of $X-\Lambda(X)\ind_B-N_{\rho^A(X),\,\rho^C(X)}$ which lies in $\acc_1\times\acc_2$.
For every $\zeta\in\R$, the allocation $(X_1(\zeta), X_2(\zeta))$ given by   
\[X_1(\zeta)=X\ind_A+(K_1-\rho^A(X)+\zeta\Lambda(X))\ind_B\]
and
\[X_2(\zeta) = (X-\rho^B(X)-\rho^C(X)-K_1-(\zeta-1)\Lambda(X))\ind_B+X\ind_C\]
is Pareto optimal. Last, we note that an optimal payoff for $X$ is given by $\rho^A(X)\ind_A+(\Lambda(X)-\rho^A(X)-\rho^C(X))\ind_B+\rho^C(X)\ind_C\in\CM$. 


\section{Law-invariant acceptance sets}\label{sec:lawinvacc}

In this section we discuss the risk sharing problem for law-invariant acceptance sets. Throughout we fix an \textit{atomless} probability space $(\Omega,\CF,\PW)$. By $\Linfty:=L^{\infty}(\Omega,\CF,\PW)$ and $L^1:=L^1(\Omega,\CF,\PW)$ we denote the spaces of equivalence classes of bounded and $\PW$-integrable random variables, respectively. They are Banach lattices when equipped with the usual $\PW$-almost sure (a.s.) order and their natural norms $\Norm_\infty:X\mapsto\inf\{m>0\mid \PW(|X|\leq m)=1\}$ and $\Norm_1:X\mapsto\E[|X|]$. All appearing (in)equalities between random variables are understood in the a.s. sense. 
\begin{definition}\label{def:lawinv}
A subset $\CC\subset L^1$ is \textsc{$\PW$-law-invariant} if $X\in\CC$ whenever there is $Y\in\CC$ which is equal to $X$ in law under $\PW$, i.e. the two Borel probability measures $\PW\circ X^{-1}$ and $\PW\circ Y^{-1}$ on $(\R,\Borel(\R))$ agree. Given a $\PW$-law-invariant set $\emptyset\neq \CC\subset L^1$ and some other set $S\neq \emptyset$, a function $f:\CC\to S$ is called $\PW$-law-invariant if $\PW\circ X^{-1}=\PW\circ Y^{-1}$ implies $f(X)=f(Y)$. 
\end{definition}

\subsection{Existence of optimal payoffs, Pareto optima, and equilibria}
Let us specify the setting.

\textit{Model space assumptions:} Throughout this section, all agents $i\in [n]$ operate on the same model space $\CX_i=\CX\subset L^1$ consisting of equivalence classes of integrable random variables. For the sake of clarity, we will first discuss the results in the maximal case $\CX=L^1$. In Section \ref{sec:genmodel}, the results will be generalised to a wide class of model spaces $L^\infty\subset\CX\subset L^1$. 

\textit{Acceptance sets:} Each agent $i\in[n]$ deems a loss profile adequately capitalised if it belongs to a closed $\PW$-law-invariant acceptance set $\acc_i\subset L^1$ which contains a riskless payoff, i.e.
\begin{equation}\label{eq:finite_on_bd}\R\cap \acc_i\neq \emptyset.\end{equation}
As the dual space of $L^1$ may be identified with $L^\infty$, we may see the respective \textit{support functions} as mappings 
\[\sigma_{\acc_i}:L^\infty\to(-\infty,\infty],\quad Q\mapsto\sup_{Y\in\acc_i}\E[QY];\]
c.f. Appendix \ref{appendix:geom}. Due to monotonicity of the sets $\acc_i$, $\dom(\sigma_{\acc_i})\subset L^\infty_+$ holds. The reader may think of acceptance sets arising, for instance, from the Average Value at Risk (Expected Shortfall) or distortion risk measures.

\textit{Security markets:} Regarding the security markets, we require there is a linear functional $\pi:\CM\to\R$ on the global security space $\CM$ such that the individual pricing functionals are given by $\price_i=\pi|_{\CS_i}$, $i\in [n]$; the agents operate on different sub-markets $(\CS_i,\pi|_{\CS_i})$ of $(\CM,\pi)$. In particular, conditions ($\star$) and (NSA) are satisfied. Moreover, we assume

\begin{assumption}\label{assumption} 
$\pi$ is of the shape $\pi(Z)=p\E_\QW[Z]$, $Z\in\CM$, where $p>0$ is a fixed constant and $\QW=Qd\PW$, $Q\in\Linfty_+$, is a probability measure such that either
\begin{enumerate}
\item[(1)] $Q=1$, i.e. $\QW=\PW$, or
\item[(2)] $Q\in \bigcap_{i=1}^n\dom(\sigma_{\acc_i})$ and for all $0\neq N\in\CM$ such that $\E_\QW[N]=0$ there exists 
$Q^N\in \bigcap_{i=1}^n\dom(\sigma_{\acc_i})$
such that $\E[Q^NN]>0$. 
\end{enumerate}
\end{assumption}

Our assumption on the pricing functionals is quite flexible as illustrated by Example \ref{ex:lawinv2} below. Assumption \ref{assumption}(2) means the agent's view on acceptability is risk averse with respect to pricing, and that a fully leveraged security $N$ with non-trivial variance and price $0$ cannot be market acceptable at all quantities.  Recall from the introduction that assuming the individual acceptance sets $\acc_i$ to be law-invariant means that being acceptable or not is merely a statistical property of the loss profile. Mathematically, this intuition necessitates introducing the hypothetical physical measure $\PW$. Prices in the security market can, e.g., be determined by a suitable martingale measure $\QW$ though. For the remainder of this section we assume that Assumption~\ref{assumption} is satisfied.

In order to deduce the existence of optimal payoffs and Pareto optimal allocations from Proposition \ref{prop:exact2} in the case $\CX=L^1$, properness of $\Lambda$ and closedness of $\augm$ must be verified. In a first step we characterise the recession cone $0^+\CC$ of a convex law-invariant set $\CC$, which is also of independent interest. For the definition of a recession cone, we refer to Appendix \ref{appendix:geom}.

\begin{proposition}\label{prop:lawinv}
Suppose $\emptyset\neq \CC\subsetneq L^1$ is law-invariant, convex, and closed. Then $0^+\CC$ is law-invariant. If furthermore $\CC$ does not agree with one of the sets 
\begin{center}$\{X\in L^1\mid c_-\leq \E[X]\le c_+\}$\end{center}
where $-\infty\leq c_-\leq c_+\leq \infty$, then $U\in 0^+\CC$ and $\E[U]=0$ imply $U=0$.  
\end{proposition}
\begin{proof}
As $\CC$ is norm closed and convex, the Hahn-Banach Separation Theorem gives the representation
\[\CC=\{X\in L^1\mid \forall\, Q\in \dom(\sigma_\CC):~\E[QX]\leq \sigma_\CC(Q)\},\]
where $\sigma_\CC$ is the support function of $\CC$.
It is well-known that $\dom(\sigma_\CC)$ is a law-invariant, closed and convex cone in $L^\infty$. 
The law-invariance of $\dom(\sigma_\CC)$ combined with Lemma~\ref{lem:direction} shows that the recession cone $0^+\CC$ is law-invariant as well. Consequently, by \cite[Lemma 1.3]{continuity}, for any $U\in 0^+\CC$, $Q\in\dom(\sigma_\CC)$, and sub-$\sigma$-algebra $\CH\subset \CF$, we have 
\begin{equation}
\label{eq:prob}\E[U|\CH]\in 0^+\CC\quad \mbox{and}\quad \E[Q|\CH]\in \dom(\sigma_{\CC}). 
\end{equation}

Now suppose that $\CC$ does not equal some set of type $\{X\in L^1\mid c_-\leq \E[X]\le c_+\}$. Then there is a  non-constant $Q\in\mathcal \dom(\sigma_\CC)$. Further, suppose $U\in 0^+\CC$ is not constant. 
As $(\Omega,\CF,\PW)$ is non-atomic, for $k\geq 2$ large enough there is a finite measurable partition\footnote{ That is the sets are pairwisely disjoint, measurable, and their union is $\Omega$.} $\Pi:=(A_1,...,A_k)$ of $\Omega$ such that $\PW(A_j)=\frac 1k$, $j\in[k]$, and $U^*=\E[U|\sigma(\Pi)]=\sum_{i=1}^ku_i\ind_{A_i}$ and $Q^*=\E[Q|\sigma(\Pi)]=\sum_{i=1}^kq_i\ind_{A_i}$ are both non-constant. As for any permutation $\tau:[k]\to[k]$ the random variable $U_\tau^*:=\sum_{i=1}^ku_{\tau(i)}\ind_{A_i}$ has the same distribution under $\PW$ as $U^*$, $U^*_\tau\in 0^+\CC$ follows. Similarly, $Q^*_\tau:=\sum_{i=1}^kq_{\tau(i)}\ind_{A_i}\in \dom(\sigma_\CC)$. For our argument we will hence assume without loss of generality that the vectors $u$ and $q$ satisfy $u_1\leq ... \le u_k$ and $ q_1\le ...\le q_k$. 
In both chains of ineqalities, at least one inequality has to be strict. We estimate
\[\E[Q]\E[U]=\E[Q^*]\E[U^*]=\left(\frac 1 k\sum_{i=1}^kq_i\right)\cdot\left(\frac 1 k\sum_{i=1}^ku_i\right)<\frac 1 k \sum_{i=1}^kq_iu_i=\E[Q^*U^*]\leq 0,\]
where the first strict inequality is due to Chebyshev's sum inequality \cite[Theorem 43]{HLP} and $u$ and $q$ being non-constant, and the last inequality is due to $U^*\in 0^+\CC$, $Q^*\in\dom(\sigma_\CC)$, and Lemma \ref{lem:direction}. $\E[U]=0$ is hence impossible.
\end{proof}

In order to apply the preceding proposition to acceptance sets, note that they have the shape $\{X\in L^1\mid c_-\le \E[X]\le c_+\}$ if, and only if, $-\infty=c_-<c_+<\infty$. We also need the notion of \textsc{comonotone partitions of the identity}, or \textsc{comonotone functions}, i.e. functions in the set 
\begin{center}$\Lip:=\left\{\mbf f=(f_1,..., f_n):\R\to\R^n\mid f_i\tn{ non-decreasing},~\sum_{i=1}^nf_i=id_\R\right\}$.\end{center}
For $\gamma>0$, we set $\Lip_\gamma:=\{\mbf f\in\Lip\mid \mbf f(0)\in[-\gamma,\gamma]^n\}$. One easily verifies that for $\mbf f\in\Lip$ and $i\in[n]$ the coordinate function $f_i$ is Lipschitz continuous with Lipschitz constant 1. From \cite[Lemma B.1]{Svindland08} we recall the following compactness result:

\begin{lemma}\label{lem:appendix} 
For every $\gamma>0$, $\Lip_\gamma\subset(\R^n)^\R$ is sequentially compact in the topology of pointwise convergence. 
\end{lemma}

\begin{proposition}\label{A+ker closed} Under the assumptions of this section $\mathscr A_++\ker(\pi)$ is a closed and proper subset of $L^1$, and $\Lambda$ is proper and l.s.c.
\end{proposition} 
\begin{proof}
The individual acceptance sets $\acc_i$ may be used to define $\PW$-law-invariant l.s.c. monetary base risk measures $\xi_i$ by
\[\xi_i(X):=\inf\{m\in\R\mid X-m\in\acc_i\}\in(-\infty,\infty],\quad X\in L^1.\]
By \eqref{eq:finite_on_bd}, $\xi_i(Y)\in\R$ holds for all bounded random variables $Y\in L^\infty$. Recall that we set $\CL_c(f):=\{s\in S\mid f(s)\leq c\}$ for a function $f:S\to[-\infty,\infty]$ and a level $c\in\R$. The identity $\CL_c(\xi_i)=c+\acc_i$ for $c\in\R$ is easily verified.
The risk measures $\xi_i$ admit a dual representation  
\begin{equation}\label{eq:dualrepxii}\xi_i(X)=\sup_{Q\in\dom(\xi_i^*)}\E[QX]-\xi_i^*(Q),\quad X\in L^1,\end{equation} where cash-additivity implies that  
\begin{equation}\label{eq:dual:ca}
\dom(\xi_i^*)\subset \{Q\in (\Linfty)_+\mid \E[Q]=1\}\quad \tn{ and }\quad\xi_i^*(Q)=\sigma_{\acc_i}(Q),~Q\in\dom(\xi_i^*).
\end{equation} 
Moreover, the infimal convolution 
$\xi:=\Box_{i=1}^n\xi_i>-\infty$ is a $\PW$-law-invariant monetary risk measure on $L^1$ as well and $\xi^*=\sum_{i=1}^n\xi_i^*$ by Lemma \ref{lem:sumdual}. Now, by \cite[Corollary 2.7]{Svindland08}, $\xi$ is l.s.c., and for each $X\in\dom(\xi)$ there is $\mbf f\in\Lip$ such that 
\begin{equation}\label{eq:exactcom}\xi(X)=\sum_{i=1}^n\xi_i(f_i(X)).\end{equation}
Suppose now $X\in L^1$ satisfies $\xi(X)\leq 0$ and let $\mbf f$ as in \eqref{eq:exactcom}. For all $i\in[n]$ we may choose $c_i\in\R$ such that $\xi_i(f_i(X)-c_i)=\xi_i(f_i(X))-c_i\leq 0$ and $\sum_{i=1}^nc_i=0$. If $g_i:=f_i-c_i$, $g_i(X)\in\CL_0(\xi_i)=\acc_i$, $i\in[n]$. Hence, $X=\sum_{i=1}^ng_i(X)\in \sum_{i=1}^n\acc_i=\mathscr A_+$.
We have thus shown that 
\[\mathcal L_0(\xi)=\mathscr A_+.\]
As $\xi$ is l.s.c. the left-hand set (and thus also the right-hand set) is norm closed.

If Assumption \ref{assumption}(1) holds, i.e. $\pi(\cdot)=p\E[\cdot]$, in view of Proposition~\ref{prop:lawinv} either $0^+\mathscr A_+\cap\ker(\pi)=\{0\}$ or $\mathscr A_+=\mathcal L_c(\E[\cdot])$ for some $c\in \R$. In the latter case, $0^+\mathscr A_+\cap\ker(\pi)=\ker(\pi)$ is a subspace as well. By Dieudonn\'{e}'s Theorem \cite[Theorem 1.1.8]{Dieudonne}, $\mathscr A_++\ker(\pi)$ is closed. 

Suppose Assumption~\ref{assumption}(2) holds. For $N\in 0^+\mathscr A_+\cap\ker(\pi)$ we have $\E_\QW[N]=0$. If $N\neq 0$, by assumption and \eqref{eq:dual:ca} there is $Q^N\in \dom(\xi^\ast)$ such that $\E[Q^NN]>0$. Hence, Lemma~\ref{lem:direction} implies  $0^+\mathscr A_+\cap\ker(\pi)=\{0\}$. Again, Dieudonn\'e's Theorem yields closedness of $\mathscr A_++\ker(\pi)$.

For properness of $\Lambda$, let $X\in L^1$ be arbitrary. Suppose $Z\in\CM$ is such that $X-Z\in\mathscr A_+$, i.e. $\xi(X-Z)\leq 0$. If $p>0$ and $\QW\ll\PW$ are chosen as in Assumption \ref{assumption}, we infer from \eqref{eq:dual:ca}
$$0\geq \E_\QW\left[X-Z\right]-\xi^*(Q)=\E_\QW[X]-\xi^*(Q)-\frac 1 p\pi(Z),$$
which implies $\pi(Z)\geq p(\E_\QW[X]-\xi^*(Q))>-\infty$. Properness follows with the representation of $\Lambda$ given in Proposition \ref{prop:operational}(2). Lower semicontinuity of $\Lambda$ is due to Proposition \ref{prop:exact2}.
\end{proof}

We are ready to prove the existence of Pareto optimal allocations.

\begin{theorem}\label{thm:lawinv} Under the assumptions of this section all $X\in\dom(\Lambda)$ admit an optimal payoff $Z^X\in\mathcal M$. In particular, for any  $X\in\dom(\Lambda)$, there exists a Pareto optimal allocation $\mbf X$ of the shape 
\begin{equation}\label{eq:optalloc}X_i=A_i-N_i+\Lambda(X)U_i,~~A_i:=f_i(X-\Lambda(X)U+N)\in\acc_i,~~i\in[n],\end{equation}
where $U_i\in\CS_i\cap L^1_{++}$ are such that $U:=\sum_{i=1}^nU_i$ satisfies $\pi(U)=1$, $N\in\ker(\pi)$ is an $X$-dependent zero cost global security, $\mbf N\in\alloc_N^s$ is arbitrary and $\mbf f\in\Lip$ is $X$-dependent. 
\end{theorem}
\begin{proof}
By Proposition \ref{A+ker closed}, $\Lambda$ is proper and $\augm$ is closed. By Proposition \ref{prop:exact2}, every $X\in\dom(\Lambda)$ admits an optimal payoff $Z^X$ and thus a Pareto optimal allocation by Proposition \ref{prop:exact1}. 
For the concrete shape of $Z^X$ and the Pareto optimal allocation, let $\mbf U\in \prod_{i=1}^n\CS_i$ be as in the assertion. As in the proof of Proposition \ref{A+ker closed}, we may find $\mbf f\in\Lip$ such that $f_i(X-Z^X)\in \acc_i$, $i\in[n]$. As $\pi(Z^X)=\Lambda(X)$,  $N:=\Lambda(X)U-Z^X\in\ker(\pi)$. For any $\mbf N\in\alloc_N^s$ we have $\mbf Z^X:=\Lambda(X)\mbf U-\mbf N\in\alloc_{Z^X}^s$. According to Proposition \ref{prop:exact1}, 
$$\mbf f(X-Z^X)+\mbf Z^X=\mbf f(X-\Lambda(X)U+N)+\Lambda(X)\mbf U-\mbf N$$
is a Pareto optimal allocation of $X$ with $\mbf f(X-\Lambda(X)U+N)\in\prod_{i=1}^n\acc_i$.   
\end{proof}

\begin{remark}
If $n=1$, $\Lambda=\Rho$ and Theorem~\ref{thm:lawinv} in fact solves the \textit{optimal payoff problem} studied in \cite{Baes}.
\end{remark}

We now turn our attention to the existence of equilibria. Proposition~\ref{prop:existenceequilibria} in conjunction with Theorem~\ref{thm:lawinv} proves

\begin{corollary}\label{cor:lawinvequilibrium}
In the situation of Theorem~\ref{thm:lawinv}, suppose that the agent system checks \tn{(NR)}. Then for every $\mbf W\in (L^1)^n$ such that $W=\sum_{i=1}^nW_i\in\tn{int}\,\dom(\Lambda)$ there is an equilibrium $(\mbf X,\phi)$. 
\end{corollary}

Finding elements in the interior of $\dom(\Lambda)$ usually requires stronger continuity properties of the involved risk measures and is an important motivation for studying the risk sharing problem on general model spaces in Section~\ref{sec:genmodel} endowed with a stronger topology than $\Norm_1$. Given a loss $W\in L^1$, the trick is to find a suitable model space $(\CX,\Norm)$ such that $W\in\tn{int}_{\Norm}\dom(\Lambda|_{\CX})$; see, e.g., \cite{Delbaen, Liebrich, Pichler, Subgradients}.

\subsection{Upper hemicontinuity of Pareto optima and equilibrium allocations}

By Lemma \ref{lem:contsel} there is a continuous selection $\Psi:\CM\to\prod_{i=1}^n\CS_i$ of $\CM\ni Z\mapsto \alloc_Z^s$. Hence, the correspondence $\widehat{\mathcal P}:L^1\two(L^1)^n$ mapping $X$ to Pareto optimal allocations of shape \eqref{eq:optalloc} such that, additionally, 
the security allocation of $N\in\ker(\pi)$ is given by $\Psi(N)$\footnote{ Recall that the $\mbf N$ in \eqref{eq:optalloc} can be chosen arbitrarily.} has non-empty values on $\dom(\Lambda)$ by Theorem \ref{thm:lawinv}. Although it might be the case that not all Pareto optimal allocations of $X\in\dom(\Lambda)$ are elements of $\widehat{\mathcal P}(X)$, $\widehat{\mathcal P}$ has the advantage of being upper hemicontinuous on the interior of the domain of $\Lambda$. 
\begin{theorem}\label{thm:lawinv2}
In the situation of Theorem \ref{thm:lawinv} suppose $\mathscr A_+$ does not agree with one of the level sets $\{X\in L^1\mid\E[X]\le c\}$, $c\in\R$. Then $\widehat{\mathcal P}$ is upper hemicontinuous at every continuity point $X\in\dom(\Lambda)$ of $\Lambda$ and, \tn{a fortiori}, on $\tn{int}\,\dom(\Lambda)$. 
\end{theorem}
\begin{proof}
We start with any sequence $(X^k)_{k\in\N}\subset\tn{int}\,\dom(\Lambda)$ that converges to $X\in\tn{int}\,\dom(\Lambda)$. For all $k\in\N$ let $\mbf X^k=(X^k_i)_{i\in[n]}\in\widehat{\mathcal P}(X^k)$. By Appendix \ref{appendixC}, it is enough to show that there is a subsequence $(k_\lambda)_{\lambda\in\N}$ and an allocation $\mbf X\in\widehat{\CP}(X)$ such that $\mbf X^{k_\lambda}\to\mbf X$ coordinatewise for $\lambda\to\infty$. To this end, we first recall the construction of $\mbf X^k$, $k\in\N$: There are sequences $(N^k)_{k\in\N}\subset\ker(\pi)$ and $(\mbf f^k)_{k\in\N}\subset\Lip$ such that 
\begin{itemize}
\item $A^k_i:=f^k_i(X^k-\Lambda(X^k)U+N^k)\in\acc_i$, $i\in[n]$;
\item $\mbf X^k=\mbf A^k+\Lambda(X^k)\mbf U-\mbf N^k,$
where $\mbf N^k=\Psi(N^k)$. 
\end{itemize}
We will establish in three steps that $(N^k)_{k\in\N}$ and $(\mbf f^k)_{k\in\N}$ lie in suitable relatively sequentially compact sets, which will allow us to choose the convergent subsequence.

First, as $\Lambda$ is continuous on $\tn{int}\,\dom(\Lambda)$ by \cite[Corollary 2.5]{eketem}, $(X^k-\Lambda(X^k)U)_{k\in\N}$ is a bounded sequence.\\
The second step is to prove that $(N^k)_{k\in\N}$ is a norm bounded sequence as well. We assume for contradiction we can select a subsequence $(k_\lambda)_{\lambda\in\N}$ such that $1\leq\|N^{k_\lambda}\|_1\uparrow\infty$. Using compactness of the unit sphere in the finite-dimensional space $\ker(\pi)$ and potentially passing to another subsequence, we may furthermore assume   
$$\frac 1{\|N^{k_\lambda}\|_1}N^{k_\lambda}\rightarrow N^*\in\ker(\pi)\backslash\{0\},~\lambda\to\infty,$$
Let $Y\in\mathscr A_+$ be arbitrary and note that
$$Y+N^*=\lim_{\lambda\to\infty}(1-\|N^{k_\lambda}\|_1^{-1})Y+\|N^{k_\lambda}\|_1^{-1}\left(X^{k_\lambda}-\Lambda(X^{k_\lambda})U+N^{k_\lambda}\right)\in\mathscr A_+,$$
as the latter set is closed and convex and the sequence $\left(X^{k_\lambda}-\Lambda(X^{k_\lambda})U\right)_{\lambda\in\N}$ is norm bounded. Hence, $N^*\in 0^+\mathscr A_+\cap\ker(\pi)$ which is trivial by Assumption \ref{assumption} and Proposition~\ref{prop:lawinv}, leading to the desired \textsc{contradiction}. $(N^k)_{k\in\N}$ has to be bounded and $\{N^k\mid k\in\N\}\subset \ker(\pi)$ is relatively (sequentially) compact by the finite dimension of the latter space. \\
In a third step, we establish relative sequential compactness for $\{\mbf f^k\mid k\in\N\}$. To this end, recall the definition of the monetary risk measures $\xi_i$ in the proof of Proposition \ref{A+ker closed}. We assert that $\xi_i^*(1)<\infty$ holds for all $i\in[n]$. Indeed, the dual conjugate $\xi_i^*$ is law-invariant function and thus dilatation monotone: for all $Q\in L^\infty$ and every sub-$\sigma$-algebra $\CH\subset\CF$, we have $\xi_i^*(\E[Q|\CH])\le \xi_i^*(Q)$. By choosing $\CH:=\{\emptyset,\Omega\}$, $\xi_i^*(1)=\inf_{Q\in\dom(\xi_i^*)}\xi^*(Q)=-\xi_i(0)\in\R$ follows. Now fix $k\in\N$ and let $I:=\{i\in[n]\mid f_i^k(0)>0\}$ and $J:=[n]\backslash I$. If $I$ is empty, $f_i^k(0)=0$ has to hold for all $i\in[n]$. Now suppose we can choose $i\in I$. We abbreviate $W^k:=X^k-\Lambda(X^k)U+N^k$ and estimate
\begin{align*}-\E[|W^k|]&\leq -\E[|f^k_i(W)-f^k_i(0)|]\leq \E[f^k_i(W^k)-f^k_i(0)]\\
&\leq \xi_i(f_i^k(W^k))+\xi_i^*(1)-f_i^k(0)\leq \xi_i^*(1)-f^k_i(0),\end{align*}
where we used that $A^k_i=f^k_i(W^k)\in\acc_i.$ 
Hence, 
\begin{equation}\label{eq:herechange1}\forall i\in I:~|f^k_i(0)|\leq \xi_i^*(1)+\|W^k\|_1.\end{equation}
If $j\in J$, we obtain from the requirement $f^k_1+...+f^k_n=id_{\R}$
\[|f_j^k(0)|=-f_j^k(0)\leq -\sum_{i\in J}f_i^k(0)=\sum_{i\in I}f_i^k(0)\leq \sum_{i\in [n]}\xi_i^*(1)+n\|W^k\|_1=:\gamma_k.\]
Thus, $\mbf f^k\in\Lip_{\gamma_k}$. As the bound $\gamma_k$ depends on $k$ only in terms of $\|W^k\|_1$ which is uniformly bounded in $k$ by the first and the second step, $\gamma:=\sup_{k\in\N}\gamma_k<\infty$ and $(\mbf f^k)_{k\in\N}\subset\Lip_\gamma$.

After passing to subsequences two times, we can find a subsequence $(k_\lambda)_{\lambda\in\N}$ such that 
\begin{itemize}
\item $\ker(\pi)\ni N:=\lim_{\lambda\to\infty}N^{k_\lambda}$ exists and thus $\Psi(N^{k_\lambda})\to\Psi(N)$ for $\lambda\to\infty$.  
\item for a suitable $\mbf f\in\Lip_\gamma$ it holds that $\max_{i\in[n]}|f^{k_\lambda}-f|\to 0$ pointwise for $\lambda\to\infty$, c.f. Lemma \ref{lem:appendix}. 
\end{itemize}
It remains to show that $\left(f^i(X-\Lambda(X)U+N)+\Lambda(X) U_i+\Psi(N)_i\right)_{i\in[n]}\in\widehat{\CP}(X)$ and that it is the limit of the subsequence of the Pareto optimal allocations chosen initially. To this end, we set $\mbf A:=\mbf f(X-\Lambda(X)U+N)$ and $g_i^{(k_\lambda)}:=f_i^{(k_\lambda)}-f_i^{(k_\lambda)}(0)$. $\PW$-a.s., the estimate 
\begin{align}\begin{split}\label{eq1}\left|A_i-A^{k_\lambda}_i\right|&\leq\left|(g_i-g^{k_\lambda}_i)(X-\Lambda(X)U+N)\right|\\
&+\left|f_i^{k_\lambda}(X-\Lambda(X)U+N)-f_i^{k_\lambda}(X^{k_\lambda}-\Lambda(X^{k_\lambda})U+N^{k_\lambda})\right|\\
&+\left|f_i(0)-f^{k_\lambda}_i(0)\right|\end{split}\end{align}
holds. The third term vanishes for $\lambda\to\infty$. 
The first tirm vanishes in norm due to dominated convergence. 
From the estimate
\begin{align*}&\left\|\left|f_i^{k_\lambda}(X-\Lambda(X)U+N)-f_i^{k_\lambda}(X^{k_\lambda}-\Lambda(X^{k_\lambda})U+N^{k_\lambda})\right|\right\|_1\\
\leq&~\left\|X-X^{k_\lambda}-(\Lambda(X)-\Lambda(X^{k_\lambda}))U+N-N^{k_\lambda}\right\|_1,\end{align*}
we infer the second term vanishes in norm, as well.
Set $\mbf N:=\Psi(N)$. Lower semicontinuity of $\Rhoi$ --- which follows from Theorem \ref{thm:lawinv} applied in the case $n=1$ --- yields 
\[\sum_{i=1}^n\Rhoi(A_i+\Lambda(X)U_i-N_i)\leq\liminf_{\lambda\to\infty}\sum_{i=1}^n\Rhoi(A_i^{k_\lambda}+\Lambda(X^{k_\lambda})U_i-N^{k_\lambda}_i)=\liminf_{\lambda\to\infty}\Lambda(X^{k_\lambda})=\Lambda(X).\]
The definition of $\Lambda$ eventually yields that the inequality is actually an equality, i.e. 
$$\sum_{i=1}^n\Rhoi(A_i+\Lambda(X)U_i-N_i)=\Lambda(X).$$
We have proved that $(A_i+\Lambda(X)U_i-N_i)_{i\in[n]}\in\widehat{\mathcal P}(X)$ and thus upper hemicontinuity, c.f. Appendix \ref{appendixC}. 

The same proof applies if $X\in\dom(\Lambda)$ is such that $\Lambda$ is continuous at $X$.
\end{proof}

\subsection{General model spaces}\label{sec:genmodel}

The aim of this section is to demonstrate that assuming the agents to operate on the space $\CX=L^1$ does not restrict the generality of Theorems \ref{thm:lawinv} and \ref{thm:lawinv2} and Corollary~\ref{cor:lawinvequilibrium}. Indeed $\CX$ may be chosen to be any law-invariant ideal within $L^1$ with respect to the $\PW$-a.s. order falling in one of the following two categories: 
\begin{itemize}
\item[{\bf(BC)}] \textit{Bounded case:} $\CX=L^\infty$ equipped with the supremum norm $\Norm_\infty$. 
\item[{\bf(UC)}] \textit{Unbounded case:} $\Linfty\subset\CX\subset L^1$ is a $\PW$-law invariant Banach lattice endowed with an order continuous law-invariant lattice norm $\Norm$.\footnote{ As $\CX$ will be a super Dedekind complete Riesz space, this translates as the fact that whenever $X_n\downarrow 0$ in order, $\|X_n\|\downarrow 0$ holds as well.}  
\end{itemize}
In the unbounded case, one can show that the identity embeddings $L^\infty\hookrightarrow \CX\hookrightarrow L^1$ are continuous, i.e. there are constants $\kappa,K>0$ such that $\|X\|\leq \kappa\|X\|_\infty$ and $\|Y\|_1\leq K\|Y\|$ holds for all $X\in L^\infty$ and $Y\in \CX$.
Moreover, for all $\phi\in\CX^*$ there is a unique $Q\in L^1$ such that $QX\in L^1$ and $\phi(X)=\E[QX]$ hold for all $X\in\CX$. 
The reader may think here of $L^p$-spaces, $1<p<\infty$, or more generally Orlicz hearts equipped with a Luxemburg norm as for instance in \cite{Orlicz1, Delbaen, Gao}.

The following extension result is crucial for this generalisation:

\begin{lemma}\label{lem:extension}
Let $\mathcal R:=(\acc,\CS,\price)$ be a risk measurement regime on a Banach lattice $\CX$ satisfying  {\bf(BC)} or {\bf(UC)}. Suppose that $\acc$ is $\Norm$-closed, law-invariant and satisfies $\acc\cap \R\neq \emptyset$, and $\price(Z)=\E[QZ]$ for some $Q\in\dom(\sigma_\acc)\cap L^\infty$. If we set $\CB:=\cl_{\Norm_1}(\acc)$, $\mathfrak R:=(\CB,\CS,\price)$ is a risk measurement regime on $L^1$ and $\rho_{\mf R}|_{\CX}=\Rho$. 
\end{lemma}
\begin{proof}
As $Q\in\dom(\sigma_\acc)\cap L^\infty$, $\sigma_\CB(Q)=\sup_{Y\in\CB}\E[QY]=\sigma_\acc(Q)$ holds and $\sigma_\CB(Q)<\infty$. In order to verify \eqref{eq:regime} suppose $X\in L^1$ and $Z\in\CS$ are such that $X+Z\in\CB$. Then 
\[\price(Z)=\E[QZ]=\E[Q(X+Z)]-\E[QX]\leq \sigma_\CB(Q)-\E[QX]<\infty.\]
$\mathfrak R$ is a risk measurement regime on $L^1$. For the identity $\rho_{\mf R}|_\CX=\Rho$, it suffices to show $\acc=\CB\cap\CX$. The set $\acc\cap L^\infty$ is not empty by assumption and $\sigma(L^\infty,L^\infty)$-closed by \cite[Lemma 1.3]{continuity}. This settles case {\bf(BC)}. In case {\bf(UC)}, fix $X\in\acc$. By \cite[Propositions 2 \& 4(2)]{StrongFatou}, there is a sequence $(\Pi_n)_{n\in\N}$ of finite measurable partitions $\Pi_n$ of $\Omega$ such that $\acc\cap L^\infty\ni\E[X|\sigma(\Pi_n)]\to X$ in norm. We infer $\acc=\cl_{\Norm}(\acc\cap L^\infty)$. Together with $\sigma(L^\infty,L^\infty)$-closedness of $\acc\cap L^\infty$, we obtain that $\acc$ is $\sigma(\CX,L^\infty)$-closed. $\acc=\CB\cap\CX$ follows.
\end{proof}

In view of the preceding lemma, we will assume that
\begin{itemize}
\item each individual acceptance set $\acc_i\subset\CX$ is closed, law-invariant and satisfies $\acc_i\cap\R\neq\emptyset$;
\item the security markets $(\CS_i,\price_i)$ agree with Assumption \ref{assumption}. 
\end{itemize}

For $\mbf f\in\Lip$, $i\in[n]$ and $X\in\CX$, $1$-Lipschitz continuity of $f_i$ yields $|f_i(X)|\leq |X|+|f_i(0)|\in\CX$ $\PW$-a.s. As $\CX$ is an ideal, $f_i(X)\in\CX$ holds as well; hence, $\mbf f(\CX)\subset\CX^n$, and if we plug in $X\in\CX$ in \eqref{eq:optalloc}, the resulting Pareto optimal allocation lies in $\CX^n$ because $\mbf U,\mbf N\in\CX^n$ as $\CS_i\subset \CX$ for all $i\in [n]$.

\begin{theorem}\label{thm:lawinv_final}
Let $\CX$ be a Banach lattice satisfying  {\bf(BC)} or {\bf(UC)} and assume the agent system $(\mathcal R_1,...,\mathcal R_n)$ is as described. Then Proposition~\ref{A+ker closed}, Theorems \ref{thm:lawinv} and \ref{thm:lawinv2} and Corollary \ref{cor:lawinvequilibrium} hold \tn{verbatim} when $\CX$ replaces $L^1$ and $\Norm$ replaces $\Norm_1$.  
\end{theorem}  
\begin{proof}
Let $\mathfrak R_i$ denote the extension of the risk measurement regime $\mathcal R_i$ to $L^1$ as in Lemma \ref{lem:extension}. Apply Theorem \ref{thm:lawinv} to $\rho_{\mathfrak R_1},...,\rho_{\mathfrak R_n}$ and $X\in\CX$ to obtain generalised versions of Theorem \ref{thm:lawinv} and Corollary \ref{cor:lawinvequilibrium}. This in conjunction with Proposition \ref{prop:exact2} generalises Proposition \ref{A+ker closed}. The proof of Theorem \ref{thm:lawinv2} only needs to be altered at \eqref{eq:herechange1} and \eqref{eq1}. We may replace $\|W^k\|_1$ by $K\|W^k\|$ in the first and use the order continuity of $\Norm$ in the second instance. 
\end{proof} 

For the final theorem on upper hemicontinuity of the equilibrium correspondence, recall that the finite risk measure $\Rho:L^\infty\to\R$ resulting from a risk measurement regime $\mathcal R$ on $L^\infty$ is \textsc{continuous from above} if $\Rho(X_n)\downarrow \Rho(X)$ whenever $(X_n)_{n\in\N}\subset L^\infty$ and $X\in L^\infty$ are such that $X_n\downarrow X$ a.s. 

\begin{theorem}\label{cor:equilibriumcont}
Assume that \tn{(NR)} is satisfied and that in case {\bf (BC)} $\rho_1$ is continuous from above, whereas in case {\bf(UC)} $\CX$ is reflexive. Suppose furthermore that $\mathscr A_+$ does not agree with a level set $\mathcal L_c\left(\E[\cdot]\right)$ and consider the correspondence $\mathcal E:\CX^n\two\CX^n\times\CX^*$ mapping $\mbf W$ to equilibrium allocations $(\mbf X,\phi)$ of shape
$$X_i=Y_i+\frac{\phi(W_i-Y_i)}{\phi(\tilde Z)}\tilde Z,\quad i\in[n],$$
where $\mbf Y\in\widehat{\CP}(W_1+...+W_n)$, $\tilde Z\in \check{\CS}$ with $\pi(\tilde Z)\neq 0$, and $\phi$ is a subgradient of $\Lambda$ at $W_1+...+W_n$. 
Then $\mathcal E$ is upper hemicontinuous in that $\mbf W^k\to \mbf W\in\prod_{i=1}^n\tn{int}\,\dom(\Rhoi)$, $k\to\infty$, and $(\mbf X^k,\phi_k)\in\mathcal E(\mbf W^k)$ implies the existence of a subsequence $(k_\lambda)_{\lambda\in\N}$ such that $(\mbf X,\phi):=\lim_{\lambda\to\infty}(\mbf X^{k_\lambda},\phi_{k_\lambda})\in\mathcal E(\mbf W)$. 
\end{theorem}
\begin{proof}
Let $\mbf W$ be such that $W:=\sum_{i=1}^nW_i\in\tn{int}\,\dom(\Lambda)$. From the proof of Proposition~\ref{prop:existenceequilibria} we infer that, indeed, every $(\mbf X,\phi)\in\mathcal E(\mbf W)$ is an equilibrium of $\mbf W$. For upper hemicontinuity, we shall first establish that the equilibrium prices of an approximating sequence lie in a sequentially relatively compact set in the dual $\CX^*$. We shall hence prove that there is $\eps>0$ and constants $c_1$ and $c_2$ only depending on $W$ such that, given any $X\in\CX$ with $\|X-W\|\leq\eps$ and any subgradient $\phi$ of $\Lambda$ at $X$, it holds that \begin{center}$\|\phi\|_*\leq c_1$ and $\Lambda^*(\phi)=\sum_{i=1}^n\Rhoi^*(\phi)\leq c_2.$\end{center}
As we shall elaborate later, these bound imply that all subgradients of $\Lambda$ at vectors in a closed ball around $W$ lie in a $\sigma(\CX^*,\CX)$-sequentially compact set.

In order to prove the assertion, continuity of $\Lambda$ on $\tn{int}\,\dom(\Lambda)$ (\cite[Corollary 2.5]{eketem}) allows us to choose $\eps>0$ such that $|\Lambda(W+Y)-\Lambda(W)|\leq 1$ whenever $\|Y\|\leq 2\eps$. 
Let now $\delta>0$ be such that $\delta\eps+\delta\|W\|\leq \eps$ and fix $X$ such that $\|X-W\|\leq\eps$ and a subgradient $\phi$ of $\Lambda$ at $X$. Moreover, suppose $Y\in\CX$ is such that $\|Y\|\leq 1$. We obtain from the subgradient inequality 
\[\Lambda(X)+\eps\phi(Y)\leq\Lambda(X+\eps Y)\leq\Lambda(W)+1.\]
Rearranging this inequality yields
\[\|\phi\|_*=\sup_{\|Y\|\leq 1}\phi(Y)\leq \frac{\Lambda(W)+1-\Lambda(X)}{\eps}\leq \frac 2{\eps}=:c_2.\]
Moreover, 
\begin{align*}\Lambda(X)&=\phi(X)-\Lambda^*(\phi)=\frac 1{1+\delta}\left(\phi((1+\delta)X)-\Lambda^*(\phi)\right)-\frac{\delta}{1+\delta}\Lambda^*(\phi)\\
&\leq \frac 1{1+\delta}\Lambda((1+\delta)X)-\frac{\delta}{1+\delta}\Lambda^*(\phi).\end{align*}
By rearranging this inequality we obtain 
\[\sum_{i=1}^n\Rhoi^*(\phi)=\Lambda^*(\phi)\leq \frac 1 {\delta}\Lambda((1+\delta)X)+\frac{1+\delta}\delta\Lambda(X)\leq \frac{2+\delta}{\delta}-\Lambda(W)=:c_1,\]
where we have used $\|(1+\delta)X-W\|\leq 2\eps$ following from the choice of $\delta$. 

Now consider a sequence  $(\mbf W^k)_{k\in\N}\subset\prod_{i=1}^n\tn{int}\,\dom(\Rhoi)$ such that $W^k_i\to W_i$, $k\to\infty$, holds for all $i\in[n]$. Without loss of generality, we may assume that $W^k:=W_1^k+...+W_n^k$ lies in the ball around $W$ with radius $\eps$.
For each $k\in\N$ assume that $(\mbf X^k,\phi_k)\in\mathcal E(\mbf W^k)$, $k\in\N$. We set 
$$X_i^k=Y_i^k+\frac{\phi_k(W_i^k-Y_i^k)}{\phi(\tilde Z)}\tilde Z,\quad i\in[n].$$ As $\mbf Y^k\in\widehat{\mathcal P}(W^k)$ and $W^k\to W$, $k\to\infty$, we may assume, after passing to a subsequence, that $\mbf Y^k\to\mbf Y\in \widehat{\CP}(W)$ by Theorem \ref{thm:lawinv2}.

We shall now select a convergent subsequence $(\phi^k)_{k\in\N}$. 
In case {\bf(BC)}, we conclude from \cite[Proposition 3.1(iii)]{Liebrich} and Lemma \ref{lem:sumdual} that $\dom(\Lambda^*)\subset\dom(\rho_1^*)\subset L^1$, which implies that all subgradients $\psi$ of $\Lambda$ have the shape $\psi=\E[\bar Q\,\cdot]$ for a unique $\bar Q\in L^1_+$. Hence, the equilibrium prices are given by $\phi_{k}=\E[Q_{k}\cdot]$ for a unique $Q_{k}\in L^1(\Omega,\CF,\PW)_+$. Moreover, all subgradients $Q_{k}$ lie in the $\sigma(L^1,\Linfty)$-compact set $\mathcal L_{c_1}(\rho_1^*)$. We may invoke the Eberlein-\v{S}mulian Theorem \cite[Theorem 6.34]{Ali} to find a subsequence $(k_\lambda)_{\lambda\in\N}$ such that $Q_{k_\lambda}\to Q\in L^1$ weakly, or equivalently $\phi_{k_\lambda}\to\phi=\E[Q\,\cdot]$ in $\sigma(\CX^*,\CX)$. In case {\bf (UC)}, reflexivity of $\CX$, the Banach-Alaoglu Theorem and the bounds above imply the existence of a sequentially relatively compact set $\Gamma$ such that $\phi\in\Gamma$ whenever $\|X-W\|\leq\eps$ and $\phi$ is a subgradient of $\Lambda$ at $X$. Hence there is a $\sigma(\CX^*,\CX)$-convergent subsequence $(\phi_{k_\lambda})_{\lambda\in\N}$.

Consequently, in both cases,
$$\phi_{k_\lambda}(W_i^{k_\lambda}-Y_i^{k_\lambda})\to \phi(W_i-Y_i),~\lambda\to\infty.$$
It remains to prove that $\phi$ is a subgradient of $\Lambda$ at $W$. But as $\Lambda^*$ is weakly* l.s.c. and $\phi_{k_\lambda}(W^{k_\lambda})\to\phi(W)$, we obtain
\begin{align*}\Lambda(W)&=\limsup_{\lambda\to\infty}\phi_{k_\lambda}(W^{k_\lambda})-\Lambda^*(\phi_{k_\lambda})=\phi(W)-\liminf_{\lambda\to\infty}\Lambda^*(\phi_{k_\lambda})\leq \phi(W)-\Lambda^*(\phi),
\end{align*}
which implies that, necessarily, $\Lambda(W)=\phi(W)-\Lambda^*(\phi)$ and $\phi$ is a subgradient of $\Lambda$ at $W$.\end{proof}

\subsection{Examples}

We conclude with two examples.

\begin{example}\label{ex:lawinv1}
We consider the model space $\CX:=L^1$ on which two agents operate with acceptability criteria given by the \textit{entropic risk measure}. More precisely, we choose $0<\beta\leq \gamma$ arbitrary and define 
$$\acc_1:=\{X\in L^1\mid \xi_\beta(X)\leq 0\},\quad \acc_2:=\{X\in L^1\mid \xi_\gamma(X)\leq 0\},$$
where, for $\alpha>0$, $\xi_\alpha(X):=\frac 1{\alpha}\log\left(\E[e^{\alpha X}]\right)$, $X\in L^1$. It is well-known, c.f. \cite[Example 2.9]{Svindland08}, that 
$$\xi:=\xi_\beta\Box\xi_\gamma=\xi_{\frac{\beta\gamma}{\beta+\gamma}}.$$ 
It can be shown that for any $\alpha>0$ the set of directions of $\mathcal L_0(\xi_\alpha)$ is given by $0^+\mathcal L_0(\xi_\alpha)=-L^1_+$. 
Any probability measure $\QW\approx\PW$ with bounded Radon-Nikodym derivative $Q$ with respect to $\PW$ thus satisfies Assumption \ref{assumption} and may be used as pricing measure for securities,\footnote{ The dual conjugate $\xi_\alpha^*$ of $\xi_\alpha$ is given in terms of the relative entropy of a $\QW\ll\PW$ with respect to $\PW$, which is finite whenever $\frac{d\QW}{d\PW}\in\Linfty$.} i.e. the pricing functionals are given by $\price_i=p\E_\QW[\cdot]$, $\QW\ll\PW$ as described and $p>0$ fixed. Moreover, we choose $A\in \CF$ such that $\QW(A)=\frac 1 2$ and $\CS_1=\CM=\tn{span}\{\ind_A,\ind_{A^c}\}$, and $\CS_2=\R\cdot\ind_A$. Given these specifications, $(\mathcal R_1,\mathcal R_2)$ is an agent system. 

Note that 
$\ker(\pi)=\{N_r:=r\ind_A-r\ind_{A^c}\mid r\in\R\}$. We will now characterise $\augm$ and set, for the sake of brevity, $\alpha:=\frac{\beta\gamma}{\beta+\gamma}$. Given the characterisation of $\mathscr A_+$, $X-N_r\in\mathscr A_+$ if, and only if, $\E[e^{\alpha X}\ind_A]\cdot\E[e^{\alpha X}\ind_{A^c}]\leq\frac 14$, as then, there is a solution $r\in\R$ to
\[0\geq \frac 1 \alpha\log\left(\E[e^{\alpha(X-N_r)}]\right)=\frac 1 \alpha\log\left(e^{-\alpha r}\E[e^{\alpha X}\ind_A]+e^{\alpha r}\E[e^{\alpha X}\ind_{A^c}]\right).\]
Now, for arbitrary $X\in\dom(\Lambda)=\dom(\xi_\alpha)$, we note that 
\begin{align*}\Lambda(X)&=\inf\{\pi(r\ind)\mid r\in \R,~X-r\ind\in\augm\}\\
&=\inf\left\{rp\Big|\,r\in\R,~e^{-\alpha r}\E[e^{\alpha X}\ind_A]\cdot\E[e^{\alpha X}\ind_{A^c}]\leq\frac 14\right\}\\
&=\frac{p}{\alpha}\left(\log\E[e^{\alpha X}\ind_A]+\log\E[e^{\alpha X}\ind_{A^c}]+2\log(2)\right).
\end{align*}
Hereafter, we choose a solution $r^*$ of  
\[e^{-\alpha r}\E[e^{\alpha(X-\Lambda(X))}\ind_A]+e^{\alpha r}\E[e^{\alpha(X-\Lambda(X))}\ind_{A^c}]=1,\]
e.g. 
\[r^*:=\log\left(\frac{2\E[e^{\alpha(X-\Lambda(X))}\ind_A]}{\sqrt{1-4\E[e^{\alpha(X-\Lambda(X))}\ind_A]\cdot\E[e^{\alpha(X-\Lambda(X))}\ind_{A^c}]}+1}\right).\]
Using the results from \cite[Example 2.9]{Svindland08}, $(\frac{\gamma}{\beta+\gamma}(X-\Lambda(X)\ind-N_{r^*}),\frac{\beta}{\beta+\gamma}(X-\Lambda(X)\ind-N_{r^*}))\in\acc_1\times\acc_2$. Consequently, the following is a Pareto optimal allocation of $X$:
\[\left(\frac{\gamma}{\beta+\gamma}(X-\Lambda(X)\ind-N_{r^*})+\Lambda(X)\ind+N_{r^*},\frac{\beta}{\beta+\gamma}(X-\Lambda(X)\ind-N_{r^*})\right)\]
\end{example}

\begin{example}\label{ex:lawinv2} Here, we choose the model space $\CX=\Linfty$ and illustrate the existence of Pareto optimal allocations for two agents with acceptance sets less similar than in Example \ref{ex:lawinv1}. To this end, we fix two parameters $\beta\in (0,1)$ and $\gamma>0$ and suppose that acceptability for agent 1 is based on the \textit{Average Value at Risk}, i.e. 
\[\acc_1=\{X\in L^\infty\mid \xi_1(X):=\tn{AVaR}_\beta(X)\leq 0\}=\{X\in L^\infty\mid \forall\,\QW\in\mathcal Q:~\E_\QW[X]\leq 0\},\]
where $\mathcal Q=\{\QW=Qd\PW\mid 0\leq Q\leq\frac 1{1-\beta}~\PW\tn{-a.s.},~\E[Q]=1\}$. The acceptance set of agent 2 is, as in Example \ref{ex:lawinv1}, given by an entropic risk measure, i.e. 
$$\acc_2:=\{X\in L^\infty\mid \xi_2(X):=\frac 1{\gamma}\log\left(\E[e^{\gamma X}]\right)\leq 0\}.$$
By \cite[Example 4.34 \& Theorem 4.52]{FoeSch},  
the support function of $\mathscr A_+=\acc_1+\acc_2$ is given by 
\[\sigma_{\mathscr A_+}(Q)=\sigma_{\acc_1}(Q)+\sigma_{\acc_2}(Q)=\begin{cases}\frac 1 {\gamma}H(\QW|\PW),~\tn{if }\QW:=Qd\PW\in\mathcal Q,\\
\infty\tn{ otherwise,}\end{cases}Q\in L^\infty,\]
where $H(\QW|\PW):=\E_\QW[\log(\frac{d\QW}{d\PW})]$ denotes the relative entropy of $\QW$ with respect to $\PW$. 
Suppose the security spaces $\CS_i$, $i=1,2$, are given as in Example \ref{ex:lawinv1} for some $A\in\CF$ with $\PW(A)\in (0,1-\beta)$. As pricing measure, we choose any $\QW^*\in\mathcal Q$ which satisfies 
\begin{equation}\label{eq:Q^*}\min_{\QW\in\mathcal Q}\QW(A)<\QW^*(A)<\max_{\QW\in \mathcal Q}\QW(A)=\frac{\PW(A)}{1-\beta}.\end{equation} As pricing rules we set $\price_i:=\E_{\QW^*}[\cdot]$, $i=1,2$, which results in 
$$\ker(\pi)=\tn{span}\{N:=\ind_A-r^*\ind_{A^c}\},\quad r^*=\frac{\QW^*(A)}{1-\QW^*(A)}.$$
Assumption \ref{assumption} is met because of \eqref{eq:Q^*}. Let $X\in L^\infty$ be any aggregated loss. Using \cite[Theorem 3]{FKM2015}, we obtain the dual representation
$$\Lambda(X)=\max_{\QW\in\widetilde{\mathcal Q}}\E_\QW[X]-\frac 1 {\gamma}H(\QW|\PW),$$
where $\widetilde{\mathcal Q}=\{\QW\in\mathcal Q\mid \QW(A)=\QW^*(A)\}$. We will now compute the right scaling factor $s\in\R$ such that $X-\Lambda(X)-sN\in \mathscr A_+$. This is the case if, and only if, we have for all $\QW\in\mathcal Q\backslash \widetilde{\mathcal Q}$
$$\E_\QW[X]-\frac 1{\gamma}H(\QW|\PW)-\Lambda(X)\leq s\E_\QW[N].$$
We obtain
\[s\begin{cases}\geq \sup_{\QW\in\mathcal Q\backslash\widetilde{\mathcal Q}:\, \QW(A)>\QW^*(A)}\frac{\E_\QW[X]-\frac 1{\gamma}H(\QW|\PW)+\Lambda(X)}{\E_\QW[N]}\\
\leq \inf_{\QW\in\mathcal Q\backslash\widetilde{\mathcal Q}:\, \QW(A)<\QW^*(A)}\frac{\frac 1{\gamma}H(\QW|\PW)+\Lambda(X)-\E_\QW[X]}{|\E_\QW[N]|},\end{cases}\]
and the bounds describe an \textit{a priori} non-empty interval. Choose any $s^*$ in this interval. Combining \cite[Proposition 3.2 \& Section 3.5]{JST}, we obtain that 
\[\left((X-\Lambda(X)-s^*N-\zeta)^+,(X-\Lambda(X)-s^*N)\wedge \zeta\right)\in\acc_1\times\acc_2.\]
for a suitable $\zeta\in\R$. Thus $(X_1,X_2)$ given by
$$X_1(\zeta)=(X-\Lambda(X)-s^*N-\zeta)^+-s^*r^*\ind_{A^c}+\Lambda(X)$$ and $$X_2(\zeta)=(X-\Lambda(X)-s^*N)\wedge \zeta+s^*\ind_A$$
is a Pareto optimal allocation of $X$.
\end{example}

\section{Optimal portfolio splits}\label{sec:regarbitrage}

In this section we study the existence of optimal portfolio splits. For a thorough discussion of this problem, we refer to Tsanakas \cite{Tsanakas}, although the problem we consider is rather akin to Wang \cite{Wang}. A financial institution holds a portfolio which yields the future loss $W$. In order to diversify the risk posed by $W$, it may consider dividing the portfolio into $n$ subportfolios $X_1,...,X_n\in\CX$, $X_1+...+X_n=W$, and transfer these subportfolios to, e.g., distinct legal entities such as subsidiaries which operate under potentially varying regulatory regimes. As observed by Tsanakas, for convex, but not positively homogeneous risk measures, \textit{without} market frictions like transaction costs risk can usually be reduced arbitrarily by introducing more subsidiaries, and hence, there is no incentive to stop this splitting procedure. However, since $n$ can be arbitrarily large, transaction costs should not be neglected in this setting, and we will study the problem of finding \textit{cost-optimal portfolio splits} under market frictions. 

To be more precise, we model the subsidiaries as a family $(\rho_i)_{i\in\N}$ of \textit{normalised} risk measures on a Fr\'echet lattice $(\CX,\peq,\tau)$ -- which entails $\Rhoi^*\geq 0$ for all $i\in\N$ -- such that the associated risk measurement regimes $(\mathcal R_i)_{i\in\N}$ check infinite supportability (SUP$_\infty$): as one and the same parent company splits the losses into $n$ subportfolios, assuming that, for each $n\in\N$, the set of subsidiaries $(\Rhoi)_{i\in[n]}$ forms an agent system satisfying (SUP) seems natural. Let further $\cost:\N\to[0,\infty)$ be a non-decreasing cost function. The transaction costs of introducing subsidiaries $i\in[n]$ and splitting a portfolio among them are given by $\cost(n)$. The condition $\lim_{n\to\infty}\cost(n)=\infty$ prevents infinite splitting. At last we introduce $\Lambda_n(X):=\inf_{\mbf X\in\alloc_X}\sum_{i=1}^n\Rhoi(X_i)$, $X\in\CX$, the usual risk sharing functional associated to $(\mathcal R_1,...,\mathcal R_n)$. Note that for all $X\in\CX$, $n\in\N$, and every $\mbf X\in\CX^n$ with $\sum_{i=1}^nX_i=X$, the estimate $\sum_{i=1}^n\Rhoi(X_i)=\sum_{i=1}^n\Rhoi(X_i)+\rho_{n+1}(0)\geq \Lambda_{n+1}(X)$
holds, which entails $\Lambda_n(X)\geq \Lambda_{n+1}(X)$, $n\in \N$. In this setting, optimal portfolio splits exist if each $\Lambda_n$ is exact on $\dom(\Lambda_n)$:

\begin{theorem}\label{thm:existenceoptsplit}
Suppose $(\mathcal R_i)_{i\in\N}$ is a sequence of risk measurement regimes on a Fr\'echet lattice $\CX$ which checks \tn{(SUP$_\infty$)} and results in all $\Rhoi$ being normalised. Moreover, assume that $\Lambda_n$ is exact on $\dom(\Lambda_n)$ for all $n\in\N$ and let $W\in\sum_{i=1}^m\dom(\Rhoi)$ for some $m\in\N$. Then there is $(n_*,X_1,...,X_{n_*})$, where $n_*\in\N$ and $X_1+...+X_{n_*}=W$ which is a solution of 
\begin{equation}\label{eq:optsplit}\sum_{i=1}^n\Rhoi(X_i)+\cost(n)\to\min\quad\text{subject to }X_1+...+X_n=W, \, n\in\N.\end{equation}
\end{theorem}
\begin{proof} Note that (SUP$_\infty$) can be rewritten as \begin{equation}\label{eq:joint}\exists\,\phi_0\in\bigcap_{i=1}^\infty\dom(\Rhoi^*):~\sum_{i=1}^\infty\Rhoi^*(\phi_0)<\infty.\end{equation}
Let $m_*:=\min\{m\in\N\mid\Lambda_m(W)<\infty\}=\min\{m\in\N\mid W\in\sum_{i=1}^m\dom(\rho_i)\}<\infty$. By \eqref{eq:joint}, we have $\Lambda_n(W)\geq\phi_0(W)-\sum_{i=1}^\infty\Rhoi^*(\phi_0)>-\infty$ for all $n\geq m_*$. Thus, $\Lambda_n(W)+\cost(n)=\infty$ whenever $n<m_*$ and 
$$\liminf_{n\to\infty}\Lambda_n(W)+\cost(n)\geq \phi_0(W)-\sum_{i=1}^\infty\Rhoi^*(\phi_0)+\lim_{n\to\infty}\cost(n)=\infty.$$
Therefore, we can find $n_*\in\N$ such that $\Lambda_{n_*}(W)+\cost(n_*)=\inf_{n\in\N}\Lambda_n(W)+\cost(n)\in\R.$
Now choose an attainable allocation $\mbf X\in \CX^{n_*}$ of $X$ such that $\Lambda_{n_*}(X)=\sum_{i=1}^{n_*}\Rhoi(X_i)$ to obtain a solution to \eqref{eq:optsplit}. 
\end{proof}

\begin{corollary}\label{cor:existenceoptsplit}
Suppose $(\mathcal R_i)_{i\in\N}$ is a sequence of risk measurement regimes on a Fr\'echet lattice $\CX$ such that all $\Rhoi$ are normalised. Then the assertion of Theorem \ref{thm:existenceoptsplit} holds under both of the following conditions:
\begin{enumerate}
\item[\tn{(1)}] The risk measures $(\rho_1,...,\rho_n)$ comply with Theorem \ref{thm:lawinv_final} for each $n\in\N$ and the pricing functionals are given by $\price_i=p\E_\QW[\cdot]|_{\CS_i}$ for a $p>0$ and a probability measure $\QW\ll \PW$ with $\sup_{Y\in\acc_i}\E_\QW[Y]\leq 0$, $i\in\N$. In particular, this is satisfied by $\QW=\PW$. 
\item[\tn{(2)}] \tn{(SUP$_\infty$)} is satisfied, and for each $n\in\N$, $(\mathcal R_1,...,\mathcal R_n)$ is a polyhedral agent system. 
\end{enumerate}
\end{corollary}
\begin{proof}
(1) Let $\QW=Qd\PW$, $Q\in\Linfty_+$, be as described in the assertion. Let $i\in\N$ be arbitrary and recall the definition of the cash-additive risk measures $\xi_i$ in the proof of Proposition \ref{prop:lawinv}. By \eqref{eq:dual:ca}, $\xi_i^*(Q)\leq 0$. 
Theorem \ref{thm:lawinv_final} in the case $n=1$ yields that each $X\in\dom(\Rhoi)$ admits an optimal payoff $Z^X\in\CS_i$, i.e. $X-Z^X\in\acc_i$ and $p\E_\QW[Z^X]=\price_i(Z^X)=\Rhoi(X)$. 
Hence, 
$$\Rhoi^*(pQ)=\sup_{X\in\dom(\Rhoi)}p\E_{\QW}[X]-\Rhoi(X)=\sup_{X\in\dom(\Rhoi)}p\E_{\QW}[X-Z^X]\leq p\xi_i^*(Q)\leq 0.$$
Conversely, as $\Rhoi$ is normalised, we have $\Rhoi^*(pQ)\geq 0$. Hence, (SUP$_\infty$) holds and $\phi_0$ in \eqref{eq:joint} may be chosen as $\phi_0=p\E_\QW[\cdot]$. As in the proof of Theorem \ref{thm:lawinv2}, $\xi_i^*(1)\leq 0$ holds for all $i\in\N$. The solvability of \eqref{eq:optsplit} under (1) follows from Theorems \ref{thm:lawinv_final} and \ref{thm:existenceoptsplit}. 

(2) By Theorem \ref{A+ker poly} $\Lambda_n$ is exact on $\dom(\Lambda_n)$ for every $n\in\N$. 
\end{proof}

\appendix
\section{Technical supplements}\label{appendixA}

\subsection{The geometry of convex sets}\label{appendix:geom} Fix a non-empty convex subset $\CC$ of a locally convex Hausdorff topological Riesz space $(\CX,\peq,\tau)$ with dual space $\CX^*$.
The \textsc{support function} of $\mathcal C$ is the functional
\[\sigma_{\mathcal C}:\CX^*\to(-\infty,\infty],\quad \phi\mapsto\sup_{Y\in\mathcal C}\phi(Y).\]
The \textsc{recession cone} of $\mathcal C$ is the set 
\[0^+\mathcal C:=\{U\in \CX\mid\forall\,Y\in\mathcal C\,\forall\,k\geq 0:~Y+kU\in \mathcal C\}.\]
A vector $U$ lies in $0^+\mathcal C$ if, and only if, $Y+U\in\CC$ holds for all $Y\in\CC$. $U$ is then called a \textsc{direction} of $\CC$.
The \textsc{lineality space} of $\mathcal C$ is the vector space $\tn{lin}(\mathcal C):=0^+\mathcal C\cap(-0^+\mathcal C)$.
In the case of an acceptance set $\acc$, monotonicity implies $\dom(\sigma_\acc)\subset\CX^*_+$. If $\mathcal C$ is closed, the Hahn-Banach Separation Theorem shows that 
\begin{center}$\mathcal C=\left\{Y\in\CX\mid \forall\, \phi\in\dom(\sigma_{\CC}):~\phi(Y)\leq\sigma_{\CC}(\phi)\right\}$.\end{center}
Combining this identity with the definition of the recession cone and the lineality space yields  
\begin{lemma}\label{lem:direction}
If $\CC\subset \CX$ is closed and convex and $\mathcal J\subset\dom(\sigma_\CC)$ is such that 
\begin{center}$\CC=\{X\in\CX\mid\forall\,\phi\in\mathcal J:~\phi(X)\leq \sigma_\CC(\phi)\}$,\end{center}
then
\[0^+\CC=\bigcap_{\phi\in\mathcal J}\mathcal L_0(\phi)=\{U\in\CX\mid~\forall\, \phi\in\mathcal J:~\phi(U)\leq 0\}\quad\textit{and}\quad\tn{lin}(\acc)=\bigcap_{\phi\in\mathcal J}\ker(\phi)\]
\end{lemma}

Last we state a decomposition result for closed convex sets specific to finite-dimensional spaces. It follows from arguments in the proofs of \cite[Lemmas II.16.2 and II.16.3]{Barvinok}.

\begin{lemma}\label{lem:decomp}Let $\CC\subset\R^d$ be convex and closed and $\mathcal V:=\tn{lin}(\CC)^\perp$. If $\tn{ext}(\CC\cap \mathcal V)$ denotes the set of extreme points of $\CC\cap \mathcal V$ and $\tn{co}(\cdot)$ is the convex hull operator, $\CC$ can be written as 
\begin{center}$\CC=\tn{co}(\tn{ext}(\CC\cap \mathcal V))+0^+\CC.$\end{center}
\end{lemma}

\subsection{Infimal convolution}\label{sec:infimal:convolution}

Let $(\CX,\peq)$ be a Riesz space and suppose that functions $g_i:\CX\to(-\infty,\infty]$, $i\in[n]$, are given. The \textsc{infimal convolution} or \textsc{epi-sum} of $g_1,..., g_n$ is the function
\begin{center}$\Box_{i=1}^ng_i:\CX\to[-\infty,\infty],\quad X\mapsto\inf\left\{\sum_{i=1}^n g_i(X_i)\mid X_1,..., X_n\in \CX,\,\sum_{i=1}^nX_i=X\right\}$.\end{center}
The convolution is said to be \textsc{exact} at $X\in\CX$ if there is $X_1,..., X_n\in \CX$ with $\sum_{i=1}^nX_i=X$ such that 
\[\sum_{i=1}^ng_i(X_i)=(\Box_{i=1}^ng_i)(X).\]

\begin{lemma}\label{lem:properties}
Suppose $\CX_i\subset \CX$, $i\in [n]$, are ideals in a Riesz space $(\CX,\peq)$ such that $\CX=\sum_{i=1}^n\CX_i$. If all $g_i:\CX\to(-\infty,\infty]$ are convex, then   $\Box_{i=1}^ng_i$ is convex. If $g_i$ is monotone on $\CX_i$ with respect to $\peq$ for all $i\in[n]$, i.e. $X,Y\in\CX_i,~X\peq Y$ implies $g_i(X)\leq g_i(Y)$, and $g_i|_{\CX\backslash\CX_i}\equiv\infty$, then $\Box_{i=1}^ng_i$ is monotone on $\CX$. 
\end{lemma}

\begin{proof}
We only prove monotonicity. Let $X,Y\in\CX$, $X\peq Y$, and let $\mbf X,\mbf Y\in\prod_{i=1}^n\CX_i$ with $\sum_{i=1}^nX_i=X$ and $\sum_{i=1}^nY_i=Y$. We thus have 
$0\peq Y-X=|Y-X|\peq \sum_{i=1}^n|Y_i-X_i|$. By the Riesz space property of $\CX$ and the Riesz Decomposition Property (c.f. \cite[Section 8.5]{Ali}), there is a vector $\mbf Z\in(\CX_+)^n$ such that $Y-X=\sum_{i=1}^nZ_i$ and such that $Z_i=|Z_i|\peq|Y_i-X_i|$, $i\in[n]$. $\CX_i$ being an ideal yields that in fact $\mbf Z\in\prod_{i=1}^n\CX_i$. By monotonicity of $g_i$ on $\CX_i$, $i\in[n]$, we obtain
\[(\Box_{i=1}^ng_i)(X)\leq\sum_{i=1}^n g_i(Y_i-Z_i)\leq \sum_{i=1}^ng_i(Y_i).\]
As $(\Box_{i=1}^ng_i)(Y)=\inf\{\sum_{i=1}^ng_i(Y_i)\mid \mbf Y\in\prod_{i=1}^n\CX_i\}$ by the assumption $g_i|_{\CX\backslash\CX_i}\equiv\infty$, infimising over suitable $\mbf Y$ on the right-hand side proves the assertion.
\end{proof}

Note that the risk sharing functional satisfies $\Lambda=\Box_{i=1}^ng_i$, where $g_i(X)=\Rhoi(X)$ if $X\in \CX_i$ and $g_i(X)=\infty$ otherwise, $X\in\CX$. These functions $g_i$ inherit convexity on $\CX$ and monotonicity on $\CX_i$ from $\Rhoi$.

\begin{lemma}\label{lem:sumdual}
Given a topological Riesz space $(\CX,\peq,\tau)$ and proper functions $g_i:\CX\to(-\infty,\infty]$, $i\in[n]$, the following identities hold:
\begin{center}$(\Box_{i=1}^ng_i)^*=\sum_{i=1}^ng_i^*$\quad and\quad 
$\dom((\Box_{i=1}^ng_i)^*)=\bigcap_{i=1}^n\dom(g_i^*).$\end{center}
\end{lemma}

\subsection{Correspondences}\label{appendixC}

Given two non-empty sets $A$ and $B$, a map $\Gamma: A\to 2^B$ mapping elements of $A$ to subsets of $B$ is called a \textsc{correspondence} and will be denoted by $\Gamma: A\two B$.  Assume now that $(\CX,\tau)$ and $(\mathcal Y,\sigma)$ are topological spaces, and let $\Gamma:\CX\two\mathcal Y$ be a correspondence.

A continuous function $\Psi:\CX\to\mathcal Y$ is a \textsc{continuous selection} for $\Gamma$ if $\Psi(x)\in\Gamma(x)$ holds for all $x\in\CX$. 

If $(\CX,\sigma)$ is first countable, $\Gamma$ is \textsc{upper hemicontinuous} at $x\in\CX$ if,  whenever $(x_k)_{k\in\N}$ is a sequence $\sigma$-convergent to $x$ and $(y_k)_{k\in\N}\subset \CY$ is such that $y_k\in\Gamma(x_k)$, $k\in\N$, there is a limit point $y\in\Gamma(x)$ of $(y_k)_{k\in\N}$. If both topological spaces are first countable, $\Gamma$ is \textsc{lower hemicontinuous} at $x\in\CX$ if, whenever $(x_k)_{k\in\N}$ is a sequence $\sigma$-convergent to $x$ and $y\in\Gamma(x)$, there is a subsequence $(k_\lambda)_{\lambda\in\N}$ and $y_\lambda\in\Gamma(x_{k_\lambda})$, $\lambda\in\N$, such that $y_{\lambda}\to y$ with respect to $\tau$ as $\lambda\to\infty$.\footnote{ Our definitions of lower and upper hemicontinuity are special cases of the more general notions suited to our purposes, c.f. \cite[Theorems 17.20 \& 17.21]{Ali}.} An example of a lower hemicontinuous correspondence relevant for our investigations is the security allocation map 
$\alloc^s_{\cdot}:\CM\ni Z\mapsto\alloc_Z\cap\prod_{i=1}^n\CS_i.$

\begin{lemma}\label{lem:contsel}
The correspondence $\alloc_\cdot^s$ is lower hemicontinuous on the global security market $\CM$ and admits a continuous selection $\Psi:\mathcal M\to\prod_{i=1}^n\CS_i$ with respect to any norm on $\CM$.
\end{lemma}
\begin{proof}
Let $\langle\cdot,\cdot\rangle$ be an inner product on $\mathcal M$. Set $\CS_0:=\{0\}$. We claim that  there are natural numbers $0=m_0<m_1\leq ...\leq m_n$ and $Z_1,...,Z_{m_{n}}\in\bigcup_{i=1}^n\CS_i$ such that for all $i\in[n]$, it holds that $\{Z_{m_{i-1}+1},...,Z_{m_i}\}$ is an orthonormal basis of $\left\{X\in\CS_i\mid X\perp \tn{span}\{Z_1,...,Z_{m_{i-1}}\}\right\}$. 
Note that every $Z\in\mathcal M$ can be expressed as $Z=\sum_{i=1}^{m_n}\langle Z_i,Z\rangle Z_i$, hence the mapping $\Psi: Z\mapsto\alloc_Z^s$ defined by 
\begin{center}$\Psi(Z)_i:=\sum_{i=m_{i-1}+1}^{m_i}\langle Z_i,Z\rangle Z_i,\quad i\in[n],$\end{center}
is a selection of $\alloc_\cdot^s$ and continuous with respect to any norm on $\mathcal M$. Lower hemicontinuity follows immediately.
\end{proof}

\end{document}